\documentclass[a4paper,onecolumn,11pt,accepted=2025-04-25]{quantumarticle}
\pdfoutput=1
\usepackage[utf8]{inputenc}
\usepackage[english]{babel}
\usepackage[T1]{fontenc}
\usepackage{amsmath}
\usepackage{hyperref}

\usepackage{graphicx}
\usepackage{braket}
\usepackage{amsfonts}
\usepackage{amssymb}
\usepackage{amsthm}

\newcommand{\Tr}{\operatorname{Tr}}
\newcommand{\id}{\operatorname{id}}
\newcommand{\AC}{\operatorname{AC}}
\newcommand{\KMS}{\operatorname{KMS}}

\newtheorem{theorem}{Theorem}[section]

\theoremstyle{definition}
\newtheorem{definition}[theorem]{Definition}

\newtheorem{example}[theorem]{Example}

\def\a{\alpha}
\def\b{\beta}
\def\k{\kappa}
\def\r{\rho}
\def\s{\sigma} 
\def\th{\theta} \def\Th{\Theta}
\def\t{\tau}

\def\ca{{\mathcal A}}
\def\cb{{\mathcal B}}
\def\ce{{\mathcal E}}
\def\cp{{\mathcal P}}
\def\cq{{\mathcal Q}}
\def\car{{\mathcal R}}
\def\cx{{\mathcal X}}
\def\cy{{\mathcal Y}}

\begin{document}
\title[Quantum detailed balance via elementary transitions]{Quantum detailed balance via elementary transitions}
\author{Rocco Duvenhage}
\orcid{0000-0003-3893-1196}
\email{rocco.duvenhage@up.ac.za}
\author{Kyle Oerder}
\orcid{0000-0003-0683-1748}
\author{Keagan van den Heuvel}
\orcid{0000-0002-6940-0534}
\affiliation{Department of Physics, University of Pretoria, Pretoria 0002, South Africa}

\begin{abstract}
Quantum detailed balance is formulated in terms of elementary transitions, in
close analogy to 
detailed balance in a classical Markov chain
on a finite set of points. 
An elementary transition is taken to be a pure state of two copies of the quantum system, 
as a quantum analogue of an ordered pair of classical points representing a classical
transition from the first to the second point. This form of
quantum detailed balance is shown to be equivalent to standard quantum
detailed balance with respect to a reversing operation, thus providing a new
conceptual foundation for the latter.
Aspects of parity in quantum detailed balance are clarified in the process.
The connection with the Accardi-Cecchini dual and the KMS dual (or Petz recovery map) is also elucidated.
\end{abstract}

\maketitle


%

\section{Introduction}


Quantum detailed balance has been approached in a number of
different ways over the decades (for example \cite{Dir, Tol, Dem, Ag, CW, Al, KFGV, M, SQV, GL, MS, AI, FU, TKRWV, AFQ, AFQ2, RLC, Ts}),
resulting in the development of a fair amount of successful theory and applications. 
However, the mathematical definitions of quantum detailed balance have tended to be quite abstract compared to the simplicity of
detailed balance in a classical Markov chain, and their physical meaning have 
not been as direct as in the classical case. 
See for instance \cite[Section I]{RLC} and \cite[Section IV.C]{Diehl15} for discussion of this point. 
A very good introductory exposition of the type of arguments involved, typically using adjoints or duals of channels or generators of quantum Markov semigroups, can be found in \cite{V}. 
This abstract approach has made it difficult to give
clear conceptual reasons why one definition should be preferred over
another. 

This paper addresses the lack of a simple direct physical or conceptual
motivation for a 
formulation of quantum detailed balance analogous to that for detailed balance in a
classical Markov chain on a set $\cx$ of points. These points are
the (pure) states of the system being studied. 
Detailed balance for the Markov chain simply requires the net transfer of
probability between any two points to be zero at each step in time. I.e., at
each step in time, the ``flow'' of probability from any point $i$ in the set
of states $\cx$, to another $j$, is the same as in the opposite direction
from $j$ to $i$. 
For a discrete set of points, this is expressed as
\begin{equation}
	\label{klasfb}
\rho_{i} \tau_{ij} = \rho_{j} \tau_{ji},
\end{equation}
for all $i$ and $j$ in $\cx$, where $\rho_{i}$ is the probability that the system is in the state $i$ at the current time, 
while if the system does happen to be in the state $i$,
then $\tau_{ij}$ is the probability for the system to transition to state $j$
during the next step in time.

\pagebreak

Thus far there does not appear to be a similarly direct argument for the quantum case. Instead the analogous arguments are phrased in the average, or in terms of expectation values, rather than using a direct quantum version of point to point transitions as in the classical case above. 

Our aim is to give a formulation of quantum detailed balance in an exact analogy to the classical argument above.
Write the transition from $i$ to $j$ as the ordered pair $(i,j)$,
which can be thought of as a pure state in the composite system consisting of two copies of the set of states $\cx$. 
We extend this to the quantum world by taking the quantum analogue of a point to point transition in a classical system to be any pure state of the composite
system consisting of two copies of the quantum system. 
In this context these pure states will be
referred to as ``elementary transitions'' (Section \ref{AfdEO}).
Any quantum channel can be expressed in terms of such elementary transitions. 


The basic notions of detailed balance can be expressed in terms of elementary transitions in analogy to the classical case in terms of classical transitions. 
A quantum version of the classical condition (\ref{klasfb}) can therefore be given (Section \ref{AfdEofb}). 
This offers a new and more direct motivation for
quantum detailed balance. 
We'll refer to this form of quantum detailed balance as elementary transition detailed balance, abbreviated as ETDB.


In the existing literature on the other hand, much evidence has accumulated on pragmatic grounds in favour of what is known as standard quantum detailed balance with respect to a reversing operation (also called SQDB-$\theta$) as arguably the best quantum extension of detailed balance. 
See in particular the series of seminal papers
\cite{FU, FR10, FU2012, FR, FR15b, FU2012b}, as well as \cite{RLC},
which show that this condition simply works exceedingly well in a number of ways,
in terms of both its implications and its natural, though abstract, mathematical form.
In addition, \cite{RLC} also motivates SQDB-$\theta$ physically through the idea of a hidden time-reversal symmetry.
Nevertheless, none of this literature appears to provide a solid motivation for SQDB-$\theta$ from first principles or from a clear conceptual or physical foundation, akin to the simple idea contained in (\ref{klasfb}) for point to point transitions in a classical Markov chain.

What we'll show (Section \ref{AfdDuaal}) is that our new formulation of quantum detailed balance, ETDB, is somewhat surprisingly equivalent to
SQDB-$\theta$. 
This is done by expressing ETDB in terms of a dual of a channel obtained by reversing its elementary transitions, and showing that this dual is the same as the Accardi-Cecchini dual \cite{AC}, thus connecting our approach to the more abstract approach through duals underlying SQDB-$\th$ in the above mentioned literature.
We argue that this singles out SQDB-$\th$ as the
best form of quantum detailed balance from a foundational stance, due to ETDB's motivation. 
This confirms the above mentioned papers'
conclusion that SQDB-$\theta$ is in many ways the best formulation of quantum detailed balance, but now from the simple clear foundational 
point of view that lead to ETDB.

Moreover, we'll use our ETDB setting to clarify the role of the reversing operation $\theta$ in SQDB-$\theta$ (Section \ref{AfdPariteit}). 
In the literature $\theta$ is often associated specifically to parity, for example a sign change in momentum in the classical detailed balance conditions for particle systems. 
However, we show that there are two aspects to $\th$, mathematically boiling down to a simple factorization of it into two operations, one being transposition with respect to an orthonormal basis in which the density matrix of the system is diagonal. This transposition does not represent any form of parity, but is instead better viewed as a mathematical part of the duality theory involved.
The other operation in the mentioned factorization can then be viewed as representing parity.

In this paper we focus on quantum systems with finite dimensional Hilbert
spaces and discrete time. 
In this way we can place more emphasis on the conceptual aspects. 
In line with this, the Appendix treats some mathematical points related to Sections 4 and 5, to make the paper as self-contained as possible. 
It consists of a brief review of the Accardi-Cecchini (AC) dual in finite dimensions, and subsequently the closely related KMS dual (also known as the standard dual or Petz recovery map) and SQDB-$\th$. 



\section{An illustrative example}
\label{AfdInlVb}

Before we turn to the general development in Section \ref{AfdEO}, we outline the core idea of our  approach to quantum detailed balance in a simple example. To emphasize how it extends the classical case, we begin with a very simple classical Markov chain which can be thought of as two opposing cycles on three points, or states, labeled 
${1,2,3}$. 

A transition from point $i$ to point $j$ will be denoted as $( i , j )$. 
Denote the probability that this system on the three points is currently at point $i$, or in state $i$, by $\r_i$.
Write the current probability distribution as the row matrix
$$
\r
=
\left[
\begin{array}
	[c]{ccc}
	\r_1 & \r_2 & \r_3
\end{array}
\right]
$$
and let
$\t = [\t_{ij}]$
be the transition matrix, that is, 
$\t_{ij}$
is the probability that during the next step in the Markov chain, the system will transition from state $i$ to state $j$, assuming that it currently is in $i$.
The probability distribution after the next transition is then given by
$$
\s = \r \t
$$
in terms of usual matrix multiplication.

As mentioned in \eqref{klasfb}, detailed balance in a Markov chain is the condition
$\r_{i} \t_{ij} = \rho_{j} \tau_{ji},$
which simply states that the probability for each transition is the same as that of its opposite. If this is satisfied, it follows that $\s = \r$, that is, the probability distribution is invariant under the dynamics described by $\t$.

We now focus on the three transitions (1,2), (2,3) and (3,1), as well their reverse or opposite transitions (2,1), (3,2) and (1,3) respectively. We are therefore assuming that the probability for any other transitions are zero,
i.e.,
$$
\t_{11} = \t_{22} = \t_{33} = 0.
$$
There is a wide range of cases where this setup satisfies detailed balance, an obvious example being
$$
\r_1 = \r_2 = \r_3 = 1/3
\qquad\mbox{and}\qquad
\t_{ij} = \t_{ji} = 1/2
\quad\mbox{for}\quad
i \neq j.
$$
In this case each transition $(i,j)$ as well as its reverse $(j,i)$ have the same probability to occur, namely the product of the above mentioned probabilities $1/3$ and $1/2$, i.e., $1/6$.

A simple but key point in our approach is that we can view a transition 
$$
(i,j)
$$ 
as a pure state of the composite system consisting of two copies of the three point system. The action of the transition $(i,j)$ is that the reduction $i$ of the state $(i,j)$ on the left (i.e., to the first system), is taken to its reduction $j$ on the right.

If we rather consider two copies of a three level quantum system, then in Dirac notation we can write the corresponding transition, in a chosen basis
$\ket{1}$, $\ket{2}$, $\ket{3}$,
as
$$
\ket{i} \ket{j}.
$$
In this context we refer to such a state of the composite system, 
when used to represent a transition in the three level system, 
as an 
\emph{elementary transition}.

This opens up the possibility of viewing other pure states of the composite quantum system as elementary transitions of the three level system as well. This is still in analogy to the classical case, but in general includes states of the composite system not available in, or allowed by, the classical case. In particular, we'll allow entangled pure states as elementary transitions.
We are no longer restricted to merely the separable states 
$\ket{i} \ket{j}$.

In keeping with the cycle (1,2), (2,3) and (3,1) in the classical example above,
consider for instance the elementary transition 
$$
\ket{\psi} 
= 
\frac{1}{\sqrt{3}}
\left(
\ket{1}\ket{2} + \ket{2}\ket{3} + \ket{3}\ket{1}
\right)
$$
Assuming for the moment that this sensibly describes a simultaneous transition of 
$\ket{1}$ to $\ket{2}$, $\ket{2}$ to $\ket{3}$, and $\ket{3}$ to $\ket{2}$, 
then in order to have an analogue of detailed balance we need to include its reverse in the full dynamics.
Intuitively this reverse should clearly be taken as the elementary transition
$$
\ket{\psi'} 
= 
\frac{1}{\sqrt{3}}
\left(
\ket{2}\ket{1} + \ket{3}\ket{2} + \ket{1}\ket{3}
\right).
$$
For the purposes of this example, let's allow only these two elementary transitions in the dynamics of the system.

To construct the full dynamics of this system, we mix the two elementary transitions into a state of the composite system,
$$
\k 
= 
p \ket{\psi} \bra{\psi} 
+ 
(1-p) \ket{\psi'} \bra{\psi'}
$$
with $p$ interpreted as the probability of that the elementary transition
$\ket{\psi}$ 
will take place, and $1-p$ that for 
$\ket{\psi'}$.
Note that these probabilities $p$ and $1-p$ are analogous to the probabilities 
$\r_{i} \t_{ij}$
in the classical case above.

We view the state $\k$ as representing the full range of elementary transitions permitted by the dynamics of our system. 
Put differently, $\k$ is taken as a complete description of the dynamics of the system. 
Note that the two elementary transitions appearing in the dynamics of the  system are orthogonal. In this sense they are completely distinct states with no overlap, and correspondingly as elementary transitions we'll view them as being independent 
parts of the system's dynamics.

We still need to explain the quantum analogues of the probability distributions $\r$ and $\s$ above in this context though.
Heuristically we can view $\ket{\psi}$ as an action taking its reduction (as a state) on the left, to its reduction on the right, in analogy to the classical transitions. 
Expanding this idea to $\k$, we view the left reduction $\r$ of $\k$ being taken to its right reduction $\s$. 
In this example, both these density matrices are simply
the maximally mixed state
$$
\r = \s = \frac{1}{3} I_3 ,
$$
as is easily verified.

Unlike the case of a classical transition $(i,j)$, however, this reduction view does not give the full picture of the inner workings of the elementary transitions. Indeed, the reductions of the elementary transitions 
$\ket{\psi}$ and $\ket{\psi'}$
are themselves both the maximally mixed state. 
In particular, the cyclic nature of the elementary transitions is not reflected in this reduction view. Nor is the action of $\k$ on states other than the maximally mixed clear yet. 
We return to this point below.


But first, as a simple illustration of detailed balance in this setup, we naturally weight the two elementary transitions equally, that is we construct the dynamics as
$$
\k 
= 
\frac{1}{2} 
\ket{\psi} \bra{\psi} + \frac{1}{2} \ket{\psi'} \bra{\psi'} .
$$
Since we have two elementary transitions, the one being the reverse of the other, with the same probability $1/2$ to occur, we view this as a form of quantum detailed balance, which we term \emph{elementary transition detailed balance}, or ETDB.

This is of course in analogy to the six elementary transitions in the classical example appearing in pairs being each other's reverse and having the same probability 1/6.

The elementary transitions
$\ket{\psi}$ and $\ket{\psi'}$,
being entangled,
are not possible in the classical case. Moreover, despite the density matrix $\k$ having just one nonzero eigenvalue $1/2$,
it cannot be written as a mixture of orthogonal pure non-entangled states. 
Indeed, all the pure states in the Hilbert space spanned by 
$\ket{\psi}$
and
$\ket{\psi'}$
are entangled, as can be confirmed.

The invariant state involved in the case of ETDB described above, is of course the maximally mixed state, as this is the reductions of $\k$ as pointed out. 
But as mentioned, we currently lack a complete description of what $\k$ and its elementary transitions do.


However, the full action of $\k$ and its elementary transitions can be expressed in terms of completely positive maps, making the connection to the usual representation of dynamics. This is possible by using the standard connection between channels and states, the Choi-Jamio{\l}kowski duality. Indeed, in the special case 
$p = 1/2$
considered here, we can simply note that the unitary operator
$$
U
=
\left[
\begin{array}
	[c]{ccc}
	0      &   0     &  1       \\
	1       &   0     &  0      \\
	0      &     1    &  0
\end{array}
\right].
$$
leads to the elementary transition $\ket{\psi} \bra{\psi}$ through the usual Choi-Jamio{\l}kowski duality, 
$$
\ket{\psi} \bra{\psi}
=
\frac{1}{3}\sum_{ij}
\ket{i}  \bra{j}  \otimes \ket{i + 1} \bra{j + 1} 
=
\frac{1}{3}\sum_{ij}
\ket{i}  \bra{j}  \otimes ( U \ket{i} \bra{j} U^\dagger ) ,
$$
(where $ \ket{4} = \ket{1} $)
while $U^\dagger$ similarly corresponds to
$\ket{\psi'} \bra{\psi'}$.
As one may have expected, $U$ (and similarly $U^\dagger$) is a unitary operator cycling through 
$\ket{1}$, $\ket{2}$ and $\ket{3}$.
Since it is unitary, it describes a fully fledged quantum channel of the three level system corresponding to the single elementary transition $\ket{\psi}$. This is again very different from the point to point transitions in the classical case, as none of the latter act on the system's state space 
$\{ 1, 2, 3 \}$
as a whole.
Neither does the elementary transition $\ket{1}\ket{2}$ for example, but in general the quantum elementary transitions may.

Thus we now have complete descriptions of the actions of the two elementary transitions on arbitrary states, 
as well as of $\k$, which corresponds via the Choi-Jamio{\l}kowski duality to the channel given by 
$$
\ce (\eta) = \frac{1}{2} U \eta U^\dagger + \frac{1}{2} U^\dagger \eta U
$$
for all density matrices $\eta$ of the three level system.

Examples of quantum detailed balance of this cycle form, though not viewed from the elementary transition point of view, have appeared in the literature. See in particular \cite[Section 6]{AFQ}, \cite{BQ}, \cite[Section 5]{FR10} and \cite[Section 7]{FR}.



\section{The Choi-Jamio{\l}kowski duality and elementary transitions}
\label{AfdEO}

We begin by setting up a framework for elementary transitions in a quantum system and more generally between two quantum systems, as well as how to decompose a channel into elementary transitions.
Elementary transitions in the quantum case are pure states of a composite quantum system,  and can include entangled states. 
Because of this they are not as simple as classical point to point transitions. 
In the classical case one has transitions between pure states of the system, or even of two different systems. In the quantum case, however, the reductions of an elementary transition to the two systems respectively, need not be pure states.
Consequently, elementary transitions need not go from one pure state to another. 
It bears repeating that it's the elementary transition itself which is a pure state of the composite system, in analogy to the classical case.

In order to express channels in terms of elementary transitions, 
we are going to make use of the Choi-Jamio{\l}kowski (CJ) duality \cite{deP, J, C} (also see \cite{JLF} for a review). 
We focus on Choi's formulation \cite{C},
although in a more general form which will be discussed shortly. 
This section in part reviews ideas from \cite[Section II]{D} in this more general context.
It sets the stage and lays down some notation and terminology for the sequel.

\subsection{The Choi-Jamio{\l}kowski map}
\label{OndAfdCJ}

The CJ duality associates a state $\kappa$, i.e., a density matrix, of a
composite system $AB$, to a channel $\mathcal{E}$ from the system $A$ to the
system $B$. We'll refer to this as the \emph{CJ map}, mapping $\mathcal{E}$ to
$\kappa$. The CJ map in its usual formulation it is given by
\begin{equation}
	\label{spesCJ}
\k
=
\frac{1}{m}\sum_{ij}
\ket{i^A}  \bra{j^A}  \otimes \ce(\ket{i^A} \bra{j^A}  ),
\end{equation}
where $\ket{1^A}  ,...,\ket{m^A}  $ is any orthonormal basis for the Hilbert
space $H_{A}$ of the system $A$. Note that in this form $\kappa$ reduces to
the maximally mixed state of $A$. Keep in mind that a channel is a completely
positive (c.p.) trace preserving linear map. When restricted to the
density matrices of $A$, it maps to the density matrices of $B$.


We'll in fact use a more general form of the CJ map. Given a state $\rho$
of $A$, i.e., a density matrix on the Hilbert space $H_{A}$, choose an
orthonormal basis $\ket{1^A}  ,...,\ket{m^A}  $ in terms of which it is
diagonal,
\[
\r
=
\r_{1} \ket{1^A} \bra{1^A}  +...+ \r_{m} \ket{m^A} \bra{m^A}  ,
\]
and set
\[
\ket{\psi} = \sum_{i=1}^{m} \r_{i}^{1/2} \ket{i^A} \ket{i^A}
\]
where 
$\ket{i^A} \ket{i^A}$ 
is the usual shorthand for 
$\ket{i^A} \otimes \ket{i^A}$. 
That is, $\ket{\psi}$ is a purification of $\r$. The more
general form of the CJ map we need is given relative to $\r$ by
\begin{equation}
	 \label{CJ}
\k
=
\id \otimes\ce(\ket{\psi}  \bra{\psi}  )
=
\sum_{ij}
\r_{i}^{1/2} \r_{j}^{1/2} 
\ket{i^A}  \bra{j^A}  \otimes \ce( \ket{i^A}  \bra{j^A} ),
\end{equation}
where $\id$ denotes the identity map. 
Note that in this case $\k$ reduces to the states 
$\r$ of $A$ and $\s = \ce( \r )$ of $B$ respectively, 
reflecting the fact that $\k$ represents a channel that takes $\r$ to $\s$.
The reason for using (\ref{CJ}) instead of (\ref{spesCJ}) will become clear in
the Section \ref{AfdEofb}, when we start our discussion of detailed balance.

\subsection{Elementary transitions}

\label{OndAfdEO} Diagonalize the state $\k$ in (\ref{CJ}) as
\begin{equation}
\label{CJ-ontb}
\k
=
\sum_{\a=1}^{N}
p_{\a}\k_{\a}
\end{equation}
where $N\leq mn$, the $\k_{\a}$'s are orthogonal pure states of the composite
system $AB$ expressed as density matrices, and the $p_{\alpha}$'s are
probabilities with
\[
p_{\alpha}>0.
\]
Note that we of course implicitly assume we've chosen our labeling such that each $\kappa_{\alpha}$ appears
only once in (\ref{CJ-ontb}), that is
\[
\kappa_{\alpha}\neq\kappa_{\beta}
\]
when $\alpha\neq\beta$. We view the pure states $\kappa_{\alpha}$ as
representing \emph{elementary transitions} from $A$ to $B$ which constitute
the channel $\mathcal{E}$.

An elementary transition is analogous to a transition from one point in a finite classical set
$\cx = \{1,...,m\}$ 
to a point in
$\cy = \{1,...,n\}$. 
Keep in mind that a Markov chain
on a finite set of points is given by the special case $m=n$. We'll consider
the classical case further as part of the next subsection.

The CJ map (\ref{CJ}) taking $\mathcal{E}$ to $\kappa$ is indeed a duality
when $\rho$ is invertible, as the CJ map is then invertible, giving a
one-to-one correspondence between the set of channels from $A$ to $B$, and the
set of states of $AB$ reducing to $\rho$ of $A$. The derivation of this\ form
of the CJ duality is essentially the same as for the usual formulation, as it
is straightforward to verify that the inverse can be expressed as
\[
\ce( \ket{\varphi}  \bra{\chi}  )
=
( \bra{\chi} \otimes I_B )
( \r^{-1/2} \otimes I_B )
\k^{TA}
( \r^{-1/2} \otimes I_B )
( \ket{\varphi}  \otimes I_B ),
\]
for any pure states $\ket{\varphi}  ,\ket{\chi}  \in H_{A}$ of $A$, where $TA$
denotes the partial transpose with respect to the basis $\ket{1^A}
,...,\ket{m^A}  $ for $H_{A}$, and $I_{B}$ is the identity operator on $B$'s
Hilbert space.

As $\kappa$ in (\ref{CJ}) is therefore a way of representing $\mathcal{E}$, we
can view (\ref{CJ-ontb}) as a way of decomposing $\mathcal{E}$ into elementary
transitions. In the sequel (\ref{CJ-ontb}) will be referred to as a \emph{CJ
decomposition} of $\mathcal{E}$ relative to $\rho$, even when $\rho$ is not
invertible. We continue to think of $\kappa$ as a (not necessarily
faithful) representation of $\mathcal{E}$, despite the fact that the
CJ map is then not necessarily invertible. 

This is a variation on a standard decomposition of a channel (see \cite{AP}).
Note that the diagonalization of $\kappa$ is in general not unique, since one
can choose different orthonormal bases in an eigenspace of dimension greater
than one, hence a CJ decomposition of a channel is not in general unique.

Again assuming that $\rho$ is invertible, the decomposition can also be
written as
\begin{equation}
\label{ontb}
\ce
=
\sum_\alpha p_\alpha \varepsilon_\alpha,
\end{equation}
where the c.p. maps $\mathcal{\varepsilon}_{\alpha}:L(H_{A})\rightarrow
L(H_{B})$ are defined in accordance with the CJ duality as
\begin{equation}
\label{ontbEO}
\varepsilon_\a (\ket{i^A}  \bra{j^A}  )
=
( \bra{j^A}  \otimes I_{B} )
( \r^{-1/2} \otimes I_{B} )
\k_\a ^{TA}
( \r^{-1/2} \otimes I_{B} )
( \ket{i^A} \otimes I_{B} ),
\end{equation}
$L(H_A)$ being the space of linear operators from $H_A$ to itself and
similarly for $L(H_{B})$. These maps $\varepsilon_{\alpha}$ can be viewed as
the elementary transitions, but for our purposes the representation as pure
states $\kappa_{\alpha}$ is more convenient because of the more direct
comparison to classical point to point transitions below. 
Although the $\varepsilon_{\alpha}$'s are seen to be completely positive by inverting (\ref{ontbEO}) to get it in the same form as (\ref{CJ}) and applying Choi's criterion \cite{C}, 
they are not necessarily channels.
In fact, an $\varepsilon_\a$ can be trace increasing, and hence need not be a quantum operation, though $p_\a \varepsilon_\a$ clearly always is, as it is trace non-increasing because of (\ref{ontb}).

When we are interested in detailed balance in a given system, i.e., for a
channel from a system to itself, we'll of course consider two copies of the
same system $A$, that is to say $B=A$. But the more general setting in terms
of two systems above clarifies the setup for CJ duality and will also be used
in Section \ref{AfdDuaal} to obtain a general form of a certain dual of a channel.

\subsection{An interpretation of elementary transitions}

\label{OndAfdInterpVanEO}

In order to understand elementary transitions better, let's consider the
classical case. We can denote a classical transition from point $i$ in the set
$\cx = \{1,...,m\}$ 
to $k$ in 
$\cy = \{1,...,n\}$ 
as the ordered pair $(i,k)$.
Consequently such a transition can mathematically equivalently be viewed as a
pure state of the composite system $\cx\cy$ consisting of $\cx$ and $\cy$, with the
cartesian product $\cx \times \cy$ as its set of pure states.

Denote the transition probability from $i$ to $k$ by $\tau_{ik}$. We embed
this into a quantum representation by considering $m$ and $n$ dimensional
Hilbert spaces $H_A$ and $H_B$ respectively, choosing any orthonormal
bases
\[
\ket{1^A}  ,...,\ket{m^A}  \text{ \ \ and \ \ }\ket{1^B}  ,...,\ket{n^B}
\]
for them, and defining a quantum channel by
\begin{equation}
	\label{klasInbed}
\ce( \ket{i^A}  \bra{j^A} )
=
\delta_{ij}\sum_{k=1}^{n}
\tau_{jk} \ket{k^B} \bra{k^B}  ,
\end{equation}
for all $i,j=1,...,m$. It is straightforward to check that $\ce$ is
trace preserving and completely positive using Choi's criterion \cite{C}
(which is essentially part and parcel of the usual CJ duality), that is,
$\ce$ is indeed a channel.

Secondly, consider a classical probability distribution
\[
\rho_{1},...,\rho_{m}
\]
over the set $\cx$. Use this to define the density matrix
\[
\rho=\sum_{i=1}^{m}\rho_{i}\ket{i^A}  \bra{i^A}
\]
to be used in the CJ map (\ref{CJ}). It is immediate from $\mathcal{E}$'s
definition that $\mathcal{E}(\rho)$ is a diagonal matrix $\sigma$ whose
diagonal entries with respect to the basis $\ket{k^B}  $,
\[
\sigma_{k} = \mathcal{E}(\rho)_{kk} = \sum_{i}\rho_{i}\tau_{ik},
\]
indeed constitute the new classical probability distribution after one step in
the Markov chain. Zero probabilities are allowed, hence $\rho$ and the CJ map
need not be invertible.

We then obtain a CJ decomposition
\begin{equation}
\label{klasCJ}
\kappa 
= 
\sum_{ik} \rho_{i}\tau_{ik} \ket{i^A}  \bra{i^A}
\otimes
\ket{k^B}  \bra{k^B}  
= 
\sum_{ik}\rho_{i}\tau_{ik}\ket{i^A}  \ket{k^B}
\bra{i^A}  \bra{k^B}
\end{equation}
of $\mathcal{E}$, expressed in terms of the elementary transitions
\begin{equation}
	\label{KwantVoorstVanKlasEO}
\k_{(i,k)} 
=
\ket{i^A} \bra{i^A} \otimes \ket{k^B} \bra{k^B}  
=
\ket{i^A} \ket{k^B} \bra{i^A} \bra{k^B}
\end{equation}
weighted by the probabilities
\begin{equation}
	\label{klasOrgWaarsk}
p_{(i,k)} = \r_{i}\tau_{ik}.
\end{equation}
Here $\a = (i,k)$ is a more convenient descriptive  label than the labeling used in (\ref{CJ-ontb}).
This quantum representation of the classical transition $(i,k)$ as the pure
state $\ket{i^A}  \ket{k^B}  $ of the system $AB$, is a simple illustration of
pure states of $AB$ as elementary transitions.

The appearance of $\rho_{i}\tau_{ik}$ is the reason we use (\ref{CJ}) rather
than (\ref{spesCJ}). That is, we work relative to the probability distribution
being considered. Keep in mind that $\rho_{i}\tau_{ik}$ is indeed the
probability for the Markov chain to transition from $i$ to $k$ during the next
step in time, assuming that $\rho_{i}$ is the probability that it is currently
in the state $i$. Put differently, $p_{(i,k)}$ is the probability that the
elementary transition $(i,k)$ takes place.

Returning to the quantum case of Subsection \ref{OndAfdEO}, any quantum pure
state $\kappa_{\alpha}$ of $AB$ appearing in (\ref{CJ-ontb}), is analogous to
a transition $(i,k)$ of the classical composite system.
The transition $(i,k)$ corresponds to the state $\ket{i^A}  \ket{k^B}  $ of
$AB$ in the quantum setting, i.e., to the (quantum) elementary transition
displayed in (\ref{KwantVoorstVanKlasEO}). The quantum case, however, allows
for types of elementary transitions not present in the classical case, in
particular entangled $\kappa_{\alpha}$. The classical reasoning above is in
line with the natural interpretation of (\ref{CJ-ontb}) as saying that
$p_{\alpha}$ is the probability that the elementary transition $\kappa
_{\alpha}$ will take place when the channel $\mathcal{E}$ is applied.

As for the meaning of an elementary transition, one expects that
it should be analogous to a classical transition $(i,k)$ from
its reduction $i$ on the first system, to its reduction $k$ on the second
system. In the quantum case, however, the reductions of $\k_{\a}$ need
not be pure states, making a quantum elementary transition more involved.

From (\ref{ontbEO}), namely,
\[
\varepsilon_\a ( \ket{i^A} \bra{j^A} ) 
= 
\r_{i}^{-1/2} \r_{j}^{-1/2}
(\bra{j^A}  \otimes I_{B}) \kappa_\a ^{TA} (\ket{i^A}  \otimes I_{B}),
\]
it is easily confirmed that $\kappa_{\alpha}$'s reduction to $B$ is given by
\[
\kappa_{\alpha}^{B} = \varepsilon_{\alpha}(\rho),
\]
hence in terms of the reductions $\kappa_{\alpha}^{A}$ and $\kappa^{A}$ of
$\kappa_{\alpha}$ and $\kappa$ to $A$,
\[
\varepsilon_\beta ( p_\alpha \k_{\alpha}^{A} ) 
\leq
\varepsilon_\beta
\left(  
\sum_\gamma p_\gamma \k_{\gamma}^{A} 
\right)  
=
\varepsilon
_{\beta}(\kappa^{A}) 
= 
\varepsilon_{\beta}(\rho) 
= 
\kappa_{\beta}^{B},
\]
simply because $\varepsilon_{\beta}$ is a positive map and $\kappa$ reduces to
$\rho$ on $A$. In this sense the state $\kappa_{\alpha}^{A}$ of $A$, when
weighted by the probability $p_{\alpha}$, is mapped ``into'' the state
$\kappa_{\beta}^{B}$ of $B$ by the elementary transition $\varepsilon_{\beta}$.

It is straightforward to check that the same is true in the classical
case. Since the reductions of $\kappa_{(i,k)}$ to $A$ and $B$ give pure states
(namely $i$ and $j$ respectively), we however have the more precise result
\[
\varepsilon_{(i,k)} ( \rho_{i} \ket{i^A}  \bra{i^A}  ) = \ket{k^B}  \bra{k^B}
\text{\ \ \ and \ \ \ } \varepsilon_{(i,k)} ( \ket{j^A}  \bra{j^A}  ) = 0
\]
for every $i$ and any $j \neq i$, in the classical case. The quantum case
above generalizes this, as the reductions of the elementary transitions in the
quantum case are in general not pure.
As pointed out in Section \ref{AfdInlVb}, in the latter case the reductions don't reflect the inner workings or action of the elementary transition in full.


%
%


A complementary discussion regarding elementary transitions can be
found in \cite[Section II]{D}, albeit in the context of the special case
(\ref{spesCJ}) of the CJ map.

\section{Quantum detailed balance defined through elementary transitions}

\label{AfdEofb}

Our goal is to give a conceptually well motivated
definition of quantum detailed balance in terms of elementary transitions.
We consider the case
\[
m=n
\]
in Section \ref{AfdEO}. The classical detailed balance condition
can be expressed in the context of that section in terms of elementary transitions. 
That gives a clear way of extending detailed balance to
the quantum case in terms of elementary transitions in Subsection
\ref{OndAfdDefEofb}. This will be referred to as elementary transition
detailed balance, or ETDB.
As we've pointed out, the CJ decomposition leading to the elementary transitions being used is in not in general unique. Subsection \ref{OndAfdOnafhVanCJ} treats the independence of ETDB from the CJ decomposition being used.
Subsequently, in Subsection \ref{OndAfdInv}, we'll
prove invariance of the state under the channel for this form of quantum detailed balance, as an illustration of the elementary transition approach.
The state $\r$ of our quantum system need not be an invertible operator in this section, 
but it will become necessary in Subsection \ref{OndAfdDuaal&FB} when we work in terms of dualities.

\subsection{Defining ETDB}
\label{OndAfdDefEofb}

A Markov chain on the set $\cx = \{1,...,m\}$ described by the transition
probabilities $\tau_{ij}$, satisfies detailed balance relative to the
probability distribution $\rho_{1},...,\rho_{m}$ exactly when
\[
\rho_{i}\tau_{ij} = \rho_{j}\tau_{ji}
\]
for all $i,j$. Clearly then, this classical detailed balance condition simply
expresses the equality of the probabilities $p_{(i,j)} = \rho_{i}\tau_{ij}$
and $p_{(j,i)} = \rho_{j}\tau_{ji}$ of the elementary transitions
$\kappa_{(i,j)}$ and $\kappa_{(j,i)}$ respectively in the CJ decomposition
(\ref{klasCJ}) of the channel $\mathcal{E}$ relative to $\rho$. It indeed
still merely directly states that the net transfer of probability between any
two points in the state space $\cx$ is zero at each step in time, which was our
initial formulation of detailed balance in a Markov chain in the Introduction.

Returning to the quantum framework set up in Section \ref{AfdEO}, but with
\[
A = B,
\]
and keeping in mind that a CJ decomposition (\ref{CJ-ontb}) need not be
unique, we extend the notion of detailed balance as follows.

As in the classical case, where we had $\kappa_{(i,j)}$ versus $\kappa
_{(j,i)}$, we have to reverse the direction of an elementary transition in the
quantum case. This is easily done by exchanging the two copies of the system.
To do this, we are going to use the map $\mathcal{R}$ on $L(H) \otimes L(H)$ given through
\[
\mathcal{R}(X \otimes Y) = Y \otimes X
\]
for any operators $X$ and $Y$ from $H$ to itself. Here
\[
H = H_{A}
\]
is the Hilbert space of our quantum system $A$ for which we want to define
detailed balance, and $L(H)$ is the space of linear operators from $H$ to
itself as before. 

Note that for any elementary transition $\k_\a$ appearing in the CJ decomposition (\ref{CJ-ontb}), we can interpret $\car(\k_\a)$ as the elementary transition in the opposite direction.
Keeping this in mind, we can now extend the elementary transition setting for
detailed balance above to the quantum case:

\begin{definition}
\label{eofb}
The channel $\ce$ is said to satisfy \emph{elementary
transition detailed balance} (or \emph{ETDB}) relative to $\r$ if
there is a CJ decomposition of $\ce$ relative to $\r$, 
as in (\ref{CJ-ontb}), such that for
each $\a$ there is a $\b$ satisfying
\begin{equation}
	\label{eofbVwd}
\car(\k_{\a})
=
\k_{\b}\text{ \ \ and \ \ }p_{\a}
=
p_{\beta}. 
\end{equation}
\end{definition}


As $\mathcal{R}$ is its own inverse, $\kappa_{\alpha}$ and $\kappa_{\beta}$
(and consequently also $\alpha$ and $\beta$) in (\ref{eofbVwd}) determine each
other uniquely, just like $(i,j)$ and $(j,i)$ in the classical case. On the
other hand, due to the possible non-uniqueness of CJ decompositions mentioned
in Section \ref{AfdEO}, the definition of ETDB may appear to depend on which
decomposition we use. It will be seen in Subsection \ref{OndAfdOnafhVanCJ} below that in practice this is
not the case.

But first note that ETDB perfectly contains the usual classical case, since in
the latter
\[
\k_{(i,j)} = \ket{i} \bra{i} \otimes \ket{j} \bra{j} 
\]
as in (\ref{KwantVoorstVanKlasEO}), giving
\[
\car( \k_{(i,j)} ) = \k_{(j,i)},
\]
with $\b = (j,i)$ indeed being the opposite transition to $\a = (i,j)$ for each such $\a$, and 
$p_{(i,j)} = p_{(j,i)}$ in (\ref{eofbVwd}) being the classical detailed balance condition by (\ref{klasOrgWaarsk}).

Although Definition \ref{eofb} is a very direct extension of the classical
case as we have just seen, there is a more succinct way of stating it.
This alternative statement of Definition \ref{eofb} will be convenient in
the next subsection to show that ETDB does not depend on the CJ decomposition being used. 
Furthermore, it will open the door to a subsequent more abstract
characterization of ETDB in Section \ref{AfdDuaal} in terms of a certain
channel dual to $\ce$, 
leading to the equivalence of ETDB to standard
quantum detailed balance with respect to a reversing operation studied in previous literature.
We provide and prove this second statement of ETDB in the next result.


\begin{theorem}
\label{eofbKar} The channel $\ce$ satisfies ETDB relative to $\r$ if and only if
\[
\car(\k) = \k.
\]
\end{theorem}

\begin{proof}
Assuming that ETDB is indeed satisfied, it follows immediately from (\ref{CJ-ontb}) and (\ref{eofbVwd}), and the fact that $\alpha$ and $\beta$ determine each other uniquely, that
\[
\car(\k)
=
\sum_\a p_{\a} \car( \k_\a )
=
\k.
\]
Conversely, assume that $\mathcal{R}(\kappa)=\kappa$. Since $\mathcal{R}$ is
implemented through
\[
\mathcal{R}(Z) = RZR
\]
for all operators $Z$ from $H\otimes H$ to itself, 
where $R:H\rightarrow H$ is the self-adjoint unitary operator given by
$$
R\ket{\varphi} \ket{\chi} = \ket{\chi} \ket{\varphi},
$$
we have $\k R = R\k$.
It follows that the eigenspaces of $\kappa$ are invariant under $R$, hence $R$
can be restricted to an operator $K\rightarrow K$ for any eigenspace $K$ of
$\kappa$. For each eigenspace of $\kappa$ we can therefore choose an
orthonormal basis consisting of eigenvectors of $R$. Denoting the density
matrices corresponding to these eigenvectors by $\kappa_{\alpha}$ for
$\alpha=1,...,n^{2}$, and letting $p_{\alpha}$ be the eigenvalue of $\kappa$
corresponding to the eigenspace of $\kappa$ that $\kappa_{\alpha}$ is
associated to, we find
\[
\k = \sum_{\a=1}^{m^2} p_\a \k_\a.
\]
By dropping the terms with $p_{\alpha}=0$, and then relabeling, we indeed
obtain
\[
\kappa=\sum_{\alpha=1}^{N}p_{\alpha}\kappa_{\alpha}
\]
for some $N\leq m^2$ and $p_{\alpha}>0$ for all $\alpha$. As $R^2 = I$, $R$ has only
the eigenvalues $1$ and $-1$, hence we have 
$R\kappa_{\alpha}R
=
R\ket{\psi_\alpha}\bra{\psi_\alpha}  R^{\dagger}
=
\ket{\psi_\alpha}  \bra{\psi_\alpha}
=
\kappa_{\alpha}$, 
where $\ket{\psi_\alpha}  $ is the eigenvector of $R$
corresponding to $\kappa_{\alpha}$ above. That is, 
$\mathcal{R}(\kappa_{\alpha})=\kappa_{\alpha}$ 
for all $\alpha$, which is sufficient for Definition \ref{eofb} to apply. 
\end{proof}

While Definition \ref{eofb} explains the meaning of ETDB by explicitly showing
each elementary transition, similar to the classical case, the restatement of
it in Theorem \ref{eofbKar} essentially just presents this definition in a
notationally more compact form. It still simply states ETDB in terms of the
reverse of elementary transitions,
but in a joint form representing the whole channel, rather than in an
individual form as in Definition \ref{eofb}. This reverse of the channel as a
whole via $\kappa$ forms the basis for the dual of a channel to be studied in
Section \ref{AfdDuaal}.

This second statement of ETDB is therefore still analogous to the classical
condition $\rho_{i}\tau_{ij}=\rho_{j}\tau_{ji}$ involving the reverse of the
transitions $(i,j)$. Nevertheless, notice that in the proof of the converse
above, the result $\mathcal{R}(\kappa_{\alpha})=\kappa_{\alpha}$, which says
that each elementary transition in the CJ decomposition is its own reverse, is
of course not exclusively what is encountered in the classical case.
Classically we typically also have cases where $\mathcal{R}(\kappa_{\alpha
})=\kappa_{\beta}$ and $p_{\alpha}=p_{\beta}$, but $\alpha\neq\beta$,
corresponding to the transition $(i,j)$ and its reverse $(j,i)$, for $i\neq
j$. The quantum case can be written in this form as well though.

Indeed, for any eigenspace $K$ of $\kappa$ for which $R$'s restriction
$R|_{K}$ to $K$ has both the eigenvalues $1$ and $-1$, we can start with
eigenvectors $\ket{\varphi}  $ and $\ket{\chi}  $ of $R|_{K}$ with eigenvalues
$1$ and $-1$ respectively, and form two new orthonormal eigenvectors
$\ket{\psi}  $ and $\ket{\omega}  $ of $\kappa$ in $K$ by normalizing
\[
\ket{\varphi}  +\ket{\chi}  \text{ \ \ and \ \ }\ket{\varphi}  -\ket{\chi}
\]
respectively. Then indeed $\ket{\psi}  \neq\ket{\omega}  $ and $R\ket{\psi}
=\ket{\omega}  $, giving $\mathcal{R}(\kappa_{\alpha})=\kappa_{\beta}$ in
terms of $\kappa_{\alpha}=\ket{\psi}  \bra{\psi}  $ and $\kappa_{\beta
}=\ket{\omega}  \bra{\omega}  $, while $p_{\alpha}=p_{\beta}$, as they are the
same eigenvalue of $\kappa$. We can repeat this in the orthogonal complement
of $\ket{\psi}  $ and $\ket{\omega}  $ in $K$, as long as there are
eigenvectors for both $R$'s eigenvalues left in it. This outcome mirrors the
typical situation in the classical case mentioned above.

This again demonstrates that in the quantum case there is some freedom
in how the elementary transitions are chosen, whereas in the classical case
the canonical choice consists of pairs of points $(i,j)$. When embedding the
classical case into the quantum case as was done in Subsection
\ref{OndAfdInterpVanEO}, we could analogously explore the same freedom of
choice in the classical case when different elementary transitions
$\kappa_{(i,j)}$ in (\ref{klasCJ}) have the same coefficient, that is, when
$\kappa$ has eigenspaces of dimension larger that $1$. This of course happens
in particular when the detailed balance condition $\rho_{i}\tau_{ij}=\rho
_{j}\tau_{ji}$ is satisfied for at least some $i\neq j$. In the latter case
one could choose a pair of orthonormal vectors other than $\ket{i}  \ket{j}  $
and $\ket{j}  \ket{i}  $ in the subspace spanned by them to serve as
elementary transitions. They would no longer be the usual classical
transitions $(i,j)$ and $(j,i)$, but would be legitimate elementary
transitions in our broader setting.



\subsection{Independence of ETDB from the CJ decomposition}
\label{OndAfdOnafhVanCJ}

Since there may be multiple CJ decompositions of a given channel $\ce$ relative to a given state $\r$, 
we would like to verify that the definition of ETDB does not
depend on the CJ decomposition being used. This will be done below. 
A basic problem in this regard is that, unlike the classical case
(\ref{klasCJ}) where for every elementary transition in a channel $\k_{(i,j)}$ the opposite
(or reverse) elementary transition $\car( \k_{(i,j)} ) = \k_{(j,i)}$
is available to potentially also appear in the CJ decomposition, in the quantum case this need not be so. 
The reverse of some
elementary transition $\k_{\a}$ in (\ref{CJ-ontb}) may simply not be available given the choices made in (\ref{CJ-ontb}).

Let's begin by giving a simple example to illustrate this, before returning to
our general quantum setup. Consider the quantum channel $\mathcal{E}$ on the
space of $2\times2$ matrices to itself given by
\[
\ce(X) = \frac{1}{2} \Tr(X) I_{2}
\]
and the usual CJ map (i.e., relative to the maximally mixed state), giving
\[
\kappa
=
\frac{1}{4} I_2 \otimes I_2
=
\frac{1}{4}\sum_{\alpha=1}^{4}
\kappa_{\alpha},
\]
with all four (the maximum available for $m = 2$) elementary transitions having non-zero probability $1/4$.
The latter CJ decomposition can be chosen with $\kappa_{1},...,\kappa_{4}$
respectively being the projections onto the orthonormal vector states
\[
	\ket{11} = \ket{1} \ket{1}  \text{, }  \ket{22} = \ket{2} \ket{2}, 
\]
\[
	a \left( \ket{1} \ket{2} + c\ket{2} \ket{1} \right)  
	\text{, }
    b \left( \ket{1} \ket{2} - \frac{1}{c^*}\ket{2} \ket{1} \right),
\]
where $\ket{1}  $ and $\ket{2}  $ are the usual coordinate basis vectors, and $a$,
$b$ and $c$ are non-zero complex numbers. 
Then the reverses $\car(\k_3)$
and $\car(\k_4)$ of $\k_3$
and $\k_4$ do not appear in this CJ decomposition (i.e., are not equal
to any of $\kappa_{1},...,\kappa_{4}$) precisely when $c$ is not purely
imaginary and $c\neq\pm1$, as is easily verified.

For condition (\ref{eofbVwd}) (for all $\alpha$) in Definition \ref{eofb} to
be a possibility, it has to be in terms of a CJ decomposition for which the reverse $\car(\k_\alpha)$ of every elementary transition $\k_\alpha$ in the CJ decomposition
is also in principle available 
to appear in that same CJ decomposition. 
More precisely, (\ref{CJ-ontb}) must have the property that the set of $\k_\a$'s is closed under $\car$, or that we can choose an orthonormal basis for the eigenspace of $\k$ corresponding to $0$, such that the resulting extended set of $\k_\a$'s, including not just those corresponding to
$p_\a > 0$, but also to $p_\a = 0$,
is closed under $\car$.
We'll call a CJ decomposition \emph{complete} if it has this property. 
In the example above this is not the case.

If (\ref{eofbVwd}) holds for any (necessarily
complete) CJ decomposition, then it holds for all complete CJ decompositions of $\ce$ relative to $\r$,
hence ETDB is independent of which complete CJ decomposition is used in Definition
\ref{eofb}.
To see this, keep in mind that if (\ref{eofbVwd}) holds for some CJ decomposition, then
$\car(\k) = \k$. For any complete CJ decomposition (\ref{CJ-ontb}) of $\ce$ relative to $\r$, we then have $p_{\alpha}=p_{\beta}$ when 
$\car(\k_{\alpha}) = \k_{\beta}$ 
(such a $\beta$ necessarily existing for every $\alpha$ by
virtue of completeness), since the $\kappa_{\alpha}$'s are orthogonal
projections and therefore linearly independent. I.e., (\ref{eofbVwd}) indeed holds for
any complete CJ decomposition.

\subsection{Invariance}

\label{OndAfdInv}

A simple and expected result regarding ETDB is the invariance 
$\ce(\r) = \r$. 
A enlightening aspect here is that one can prove this invariance
by an argument that mirrors the classical derivation. 
This illustrates how elementary transitions and their reverses in the quantum case can be viewed and treated in a way somewhat analogous to classical transitions and
their reverses in a Markov chain. 
Unlike the
classical case where the transitions are of the form $(i,j)$, there is more
freedom of choice in the elementary transitions in the quantum case. Because
of this the elementary transitions are handled
collectively via $\k$ in the arguments below. 
We'll derive invariance in two ways, one involving $\ce$ explicitly, the other purely in terms of elementary transitions without reference to $\ce$.
The latter illustrates the advantage a purely elementary transition approach may have in some contexts.

But first consider the simple argument leading to invariance in the classical
case. Writing the classical probability distribution and the transition
probabilities as a row matrix and square transition matrix
\[
\rho=\left[
\begin{array}
[c]{rrr}%
\rho_{1} & \cdots & \rho_{n}
\end{array}
\right]  \text{ \ \ and \ \ }\tau=\left[  \tau_{ij}\right]  ,
\]
we find that the $j$'th entry of the row $\rho\tau$ is given by
\begin{equation}
(\rho\tau)_{j}=\sum_{i=1}^{n}\rho_{i}\tau_{ij}=\sum_{i=1}^{n}\rho_{j}\tau
_{ji}=\rho_{j}, \label{klasInv}
\end{equation}
assuming the detailed balance condition $\rho_{i}\tau_{ij}=\rho_{j}\tau_{ji}$,
since any row in the transition matrix adds up to $1$ (all the probability at
a point must flow somewhere during a step of the Markov chain). That is,
$\rho\tau=\rho$, which is the invariance of $\rho$ under $\tau$ in the
classical case.

Returning to the quantum setup of Subsection \ref{OndAfdDefEofb}, let's see
how the simple classical argument above is reflected in the quantum argument
in terms of elementary transitions. First, let's set it up. Note that from
(\ref{CJ}) we have
\[
( \bra{i}  \otimes I ) \k (\ket{j}  \otimes I) 
=
\r_{i}^{1/2} \r_{j}^{1/2} \ce( \ket{i}  \bra{j} ),
\]
hence
\[
\ce(\r)
=
\sum_{i}
\r_{i} \ce(\ket{i}  \bra{i}  )
=
\sum_{i}
( \bra{i}  \otimes I ) \k (\ket{i} \otimes I ),
\]
and consequently
\[
\bra{j}  \ce(\r) \ket{k}  
=
\sum_{i}
\bra{i}  \bra{j}  \k \ket{i} \ket{k} .
\]

Now, to prove invariance, note from the definition of $\mathcal{R}$ that
\[
\bra{i}  \bra{j}  \kappa\ket{i}  \ket{k}  =\bra{j}  \bra{i}  \mathcal{R}
(\kappa)\ket{k}  \ket{i}  .
\]
Assuming ETDB in Definition \ref{eofb} is satisfied, 
then as in the proof of Theorem \ref{eofbKar} it leads immediately to 
$\car(\k) = \k$, 
from which it follows that
\[
\bra{j}  \mathcal{E}(\rho)\ket{k}  
=
\bra{j}  
\left(  
\sum_{i}
( I \otimes \bra{i}  )\k (I \otimes \ket{i}  )
\right)  
\ket{k}  
=
\bra{j} \Tr_{2}(\k) \ket{k}  ,
\]
i.e., $\ce(\r) = \Tr_{2}(\k)$, 
where $\Tr_{2}$ here denotes the
partial trace with respect to the second copy of $A$. Note that this is
analogous to the step
\[
\sum_{i=1}^{n}\rho_{i}\tau_{ij}=\sum_{i=1}^{n}\rho_{j}\tau_{ji}
\]
in the classical derivation (\ref{klasInv}) of invariance above. It only
remains to calculate
\begin{align*}
\Tr_{2}(\k)  
&  =
\sum_{i}
( I \otimes \bra{i}  )
\left(  
\sum_{jk}
\rho_{j}^{1/2}\rho_{k}^{1/2} \ket{j}  \bra{k}  \otimes \ce( \ket{j}\bra{k} )
\right)  
( I \otimes\ket{i}  )\\
&  =
\sum_{jk}
\rho_{j}^{1/2}\rho_{k}^{1/2}\ket{j}  \bra{k}  
\left(  
\sum_{i}
\bra{i}  \ce( \ket{j}  \bra{k}  )\ket{i}  
\right) \\
&  =
\sum_{jk}
\rho_{j}^{1/2}\rho_{k}^{1/2}\ket{j}  \bra{k}  
\Tr
\left(
\ce(\ket{j}  \bra{k}  )
\right) \\
&  =\rho,
\end{align*}
as $\ce$ is trace preserving (being a channel). This
is analogous to the step
\[
\sum_{i=1}^{n}\rho_{j}\tau_{ji}=\rho_{j}
\]
in the classical case (\ref{klasInv}). In conclusion, we have shown that
\begin{equation}\label{Inv}
\mathcal{E}(\rho)=\rho,
\end{equation}
as required.

The derivation above is already insightful with regards to the nature of elementary transitions, 
in particular their similarity to transitions in a classical Markov chain. 
As an indication of the usefulness of a purely elementary transition point of view in certain situations, without reference to $\ce$ as above, we also point out that the invariance is almost trivially obtained from Theorem \ref{eofbKar} through the following one line argument:
\begin{equation}\label{InvKort}
\Tr_1 (\k) = \Tr_1 (\car(\k)) = \Tr_2 (\k).
\end{equation}
Here $\Tr_1$ of course denotes the partial trace over the first copy of $A$, while the second equality is immediate from the definition of $\car$ as swapping the systems. 
Given that $\k$ is set up to take $\r$ to $\s = \ce(\r)$, as in Subsection \ref{OndAfdCJ}, equation \eqref{InvKort} simply states that the reduction $\s$ of $\k$ to the second copy of $A$ (by $\Tr_1$) is the same as its reduction $\r$ to the first copy. I.e., the same invariance as obtained in \eqref{Inv}.
Keep in mind that the ETDB balance condition in Definition \ref{eofb}, 
as well as Theorem \ref{eofbKar} and its proof, don't refer directly to $\ce$ either.

The invariance result will be of use in Subsection \ref{OndAfdDuaal&FB}.

\section{Dual channels and standard quantum detailed balance}
\label{AfdDuaal}


Various forms of quantum detailed balance have been proposed and studied. One
form which has been argued to be particularly physically relevant and useful is standard quantum
detailed balance with respect to a reversing operation $\theta$, or
SQDB-$\theta$ for short. 

The papers
\cite{FU, FR10, FU2012, FR, FR15b, FU2012b} comprise the initial important work making this point.
More recently \cite{RLC} have also argued in favour of SQDB-$\theta$ through the mathematically equivalent (for a specific $\th$) notion of hidden time-reversal symmetry (HTRS), and  by demonstrating that it is broad and flexible
enough to lead to useful results and techniques for quantum systems. 
Also see \cite{BQ2, DS, DS2, DSS} for more work in this vein. 

There does not seem to be a clear origin of this detailed balance
condition in the literature, and it is considered to be part of the folklore
of quantum detailed balance. 
The first place a standard detailed balance condition appeared in print, but not with respect to a reversing operation, is probably \cite{DF}, and for the case with respect to a reversing operation it seems to be \cite{FU}.

It appears that SQDB-$\theta$ is not well motivated in the literature from a foundational or first principles point of view in the spirit of the simple meaning (\ref{klasfb}) of detailed balance in a classical Markov chain, along the lines of ETDB above.
Rather, it is motivated on the basis of its consequences and implications, or in the case of \cite{RLC}, from the concept of HTRS which is still not as direct an analogue of point to point transitions in a classical system as is ETDB. 

In this section we intend to cast light on the fundamental meaning of SQDB-$\theta$ by proving it to be equivalent to ETDB for a specific $\th$.
The relevant reversing operation $\th$ emerges automatically from our ETDB point of view to be the transposition with respect to our chosen basis in which $\rho$ is diagonal. 
For this reason we'll refer to SQDB-$\th$ with this $\th$ simply as SQDB.
The same $\th$ also emerges from HTRS in the setting of \cite{RLC}.
Transposition with respect to some basis is indeed a typical choice
made in papers on SQDB-$\theta$, but without completely convincing conceptual reasons, 
or a clear general explanation of its physical role.  

The equivalence between SQDB and ETDB then in effect motivates SQDB from our elementary transition point of view and makes its analogy to classical detailed balance much clearer. 
We consider this as strong further evidence in favour of viewing SQDB as a conceptually and physically sensible form of quantum detailed balance, 

Because of this line of argument, one sees that taking $\th$ to be the mentioned transposition does not involve any form of parity, for example negation of momentum when considering the opposite direction of a classical transition in a particle system. 
SQDB, being equivalent to ETDB, instead simply describes a quantum version of detailed balance where the the opposite direction of every elementary transition is considered in order to compare the flow of probability in the two directions, without any parity involved. 
In existing literature the reversing operation is to the contrary often assumed to necessarily be related to parity, even for this transposition. 

In the next section we'll turn to parity by extending ETDB to include parity, and subsequently showing its equivalence to SQDB-$\th$ for a general $\th$.
This will clarify how precisely parity is built into $\th$ and SQDB-$\th$.

This equivalence (with or without parity) in addition provides examples of ETDB through examples in the literature on SQDB-$\th$ and HTRS above. We'll leave the analysis of these examples through the lens of elementary transitions to later work.

%
%

In order to prove the equivalence of ETDB and SQDB, we need to develop some theory regarding dual channels in Subsection \ref{OndAfdDuaal}, 
as this is the framework in which SQDB is defined. 
Because of this the presentation in Subsection \ref{OndAfdDuaal} is necessarily somewhat more mathematical than previous sections.
Of particular relevance to SQDB is the Accardi-Cecchini (AC) dual
of a channel introduced in a much more general setting in \cite[Proposition 3.1]{AC}
(a concise review of it can also be found in \cite[Theorem 2.5]{DS2}). Some facts regarding the AC dual in our finite dimensional setting is relegated to the Appendix.

The central mathematical point to be shown in Subsection \ref{OndAfdDuaal} is that the AC dual $\ce^{\AC}$ of a channel $\ce$ is the same as a dual $\ce'$ which will be defined from the reverse $\car(\k)$ of the representation $\k$ of $\ce$ through the CJ map given by (\ref{CJ}).
The dual $\ce'$ in turn has a clear conceptual meaning in our elementary transition framework as the channel obtained from $\ce$ by replacing all elementary transitions appearing in it by their opposites. This allows ETDB to be easily expressed in terms of $\ce'$ and subsequently proven equivalent to SQDB in Subsection \ref{OndAfdDuaal&FB}.

\subsection{The dual}
\label{OndAfdDuaal}

In part for sake of generality, but more importantly for clarity, we now
return to a channel $\mathcal{E}$ from system $A$ to system $B$ as in Section
\ref{AfdEO}. We are going to develop a dual channel $\ce'$
from $B$ to $A$ which will allow an alternative characterization of ETDB in
the case $A=B$, and subsequently lead to SQDB. 
Developing $\ce'$ first for the general case of two systems will, however, give a
clearer conceptual picture of what is happening. We therefore expand on our
notation from Section \ref{AfdEofb} before we proceed to treat the dual of a channel.

Consider two Hilbert spaces $H_A$ and $H_B$ of dimensions $m$
and $n$ respectively, as in Section \ref{AfdEO}. Now define $\car$ to be the
map
\[
\car : L( H_A \otimes H_B ) \rightarrow L(H_B \otimes H_A)
\]
linearly extending the requirement
$\car(X \otimes Y) = Y \otimes X$, 
where 
$X \in L(H_{A})$ 
and
$Y \in L(H_{B})$. 

As in Subsection \ref{OndAfdDefEofb}
we can interpret $\car(\k_\a)$ as the elementary transition opposite to an elementary transition $\k_\a$ appearing in a channel $\ce$ from $A$ to $B$ 
via the CJ decomposition (\ref{CJ-ontb}). Because of this, when the density matrix $\ce(\r)$ is invertible, the resulting reverse 
$$
\car(\k)
=
\sum_{\a = 1}^{N}
p_\a \car( \k_\a )
$$ 
of $\k$ will be seen to represent a channel $\ce'$ from $B$ to $A$ in which each elementary transition appearing in $\ce$ via (\ref{CJ-ontb}) has been replaced by its opposite elementary transition. In this sense $\ce'$ is the reverse of $\ce$ on the level of elementary transitions (though not its inverse as a map).

To begin with, let $\ce$ be any positive (i.e., 1-positive) linear map from
$L(H_A)$ to $L(H_B)$. For the moment we don't yet assume that
$\mathcal{E}$ is a channel. We nevertheless set
\begin{equation}
	\label{CJweer}
\k
=
\sum_{ij}
\r_{i}^{1/2}\r_{j}^{1/2}
\ket{i^A}  \bra{j^A}
\otimes
\ce( \ket{i^A}  \bra{j^A} )
\end{equation}
in line with (\ref{CJ}), where $\r$ and 
$\ket{1^A}  ,..., \ket{m^A}$ 
are as in Subsection \ref{OndAfdCJ}.
The reverse of $\kappa$ will be denoted as
\begin{equation}
	\label{kap'}
\k'
=
\car(\k)
=
\sum_{ij}
\rho_{i}^{1/2}\rho_{j}^{1/2} \ce(\ket{i^A}  \bra{j^A}  ) \otimes \ket{i^A}  \bra{j^A} .
\end{equation}

We set
\[
\s = \ce(\r),
\]
which is positive,  $\s \geq 0$, since $\r$ and $\ce$ are, but not necessarily a state. In particular, it is selfadjoint, hence we can choose an orthonormal basis 
$\ket{1^B}  ,..., \ket{n^B} $ in which it is
diagonalized as
\[
\s = \s_1 \ket{1^B} \bra{1^B} +...+ \s_{n} \ket{n^B} \bra{n^B}  .
\]

Let's now assume that $\sigma$ is invertible, i.e.,
\[
\sigma_{k}>0
\]
for all $k$. This ensures that there is a uniquely determined linear map
$\ce'$ from $L(H_{B})$ to $L(H_{A})$ such that
\begin{equation}
	\label{kap'CJ}
\k'
=
\sum_{kl}
\sigma_{k}^{1/2}\sigma_{l}^{1/2}\ket{k^B}  \bra{l^B}
\otimes\ce'(\ket{k^B}  \bra{l^B}  ), 
\end{equation}
since the latter can be inverted as follows using (\ref{kap'}),
\begin{align}
\ce'  &  ( \ket{k^B} \bra{l^B} ) 
= 
\s_k^{-1/2}\s_l^{-1/2}(\bra{k^B} \otimes I_{A})\k'(\ket{l^B}  \otimes I_{A})
\nonumber\\
&  
=
\s_k^{-1/2}\s_l^{-1/2} \sum_{ij}
\r_{i}^{1/2}\r_{j}^{1/2} \bra{k^B}  
\ce( \ket{i^A}  \bra{j^A} ) 
\ket{l^B}  \ket{i^A}
\bra{j^A} 
\label{E'formule}\\
&  
=
\sum_{ij}
\Tr
\left(  
\ket{k^B}  \bra{l^B}  
\left(  
\sigma^{-1/2}\ce(\rho^{1/2}\ket{i^A}  \bra{j^A}  \rho^{1/2})\sigma^{-1/2}
\right)  ^{T}
\right)  
\ket{i^A}  \bra{j^A} 
\nonumber
\end{align}
where $T$ is transposition with respect to 
$\ket{1^B} ,...,\ket{n^B} $.
As opposed to $\k$ in (\ref{CJweer}) where we work relative to $\r$, in $\k'$ we work relative to $\s$, since the roles of $A$ and $B$ are swapped.
For any $Y$ in $L(H_{B})$ it follows by a few algebraic manipulations that
\begin{align}
	\label{duaal}
\ce'(Y)  
&  
=
\sum_{ij}
\Tr
\left(  
Y
\left(  
\sigma
^{-1/2}\ce(\rho^{1/2}\ket{i^A}  \bra{j^A}  \rho^{1/2})\sigma^{-1/2}
\right)  ^T
\right)  
\ket{i^A}  \bra{j^A} 
\nonumber\\
&  
= 
\r^{1/2}
\left[
\ce^\dagger(\s^{-1/2}Y^{T}\s^{-1/2}) 
\right] ^ T
\r^{1/2}, 
\end{align}
where the outer $T$ is transposition with respect to 
$\ket{1^A} ,..., \ket{m^A} $,
the inner $T$ still with respect to $\ket{1^B}  ,..., \ket{n^B} $,
and $\ce^\dagger$ is defined through 
$\Tr(X\ce^\dagger(Y)) 
= 
\Tr(\ce(X)Y)$ 
for all $X$ in $L(H_A)$ 
and $Y$ in $L(H_B)$.
This linear map $\ce'$ is called the \emph{dual} of
$\ce$ relative to $\r$, or simply the dual of $\ce$ when
$\r$ is clear from context. Since it corresponds to the reverse
$\k'$ of $\k$ via \eqref{kap'CJ}, 
in which every elementary transition $\k_\a$ constituting $\ce$ has been replaced by its opposite elementary transition $\car(\k_\a)$, we can also think of
$\ce'$ as the \emph{reverse} of $\ce$ (relative to $\r$) via elementary transitions.

In the literature reversibility of channels has been studied from the point of view of the question whether a channel exists that reverses the action of a given channel on a specified set of states; see for example \cite{Pet88, Jen, Sh}. 
Related work, also emphasizing the simplifying role that SQDB-$\th$ can play, appeared in \cite{Ts1}, while \cite{BK} focused on preservation of entanglement and other correlations.
Refer to the Appendix for a few further remarks.
However, here we constructed a specific channel $\ce'$ from an elementary transition point of view, that will connect directly to ETDB. 

Let's consider properties of $\ce'$, making additional assumptions about $\ce$ as needed. Firstly, since $\s = \ce(\r)$, it is easily verified from (\ref{duaal})
that
\[
\Tr(\ce'(Y))=\Tr(Y),
\]
that is, $\ce'$ is trace preserving. 

If $\ce$ is completely positive, then $\ce'$ is a channel, as can be verified
 by the following argument:
We only need to show that $\ce'$ is c.p., as we already know
that it is trace preserving. Because $\ce$ is completely positive, we
know from 
$\k = \id \otimes \ce( \ket{\psi}  \bra{\psi} )$ 
in
(\ref{CJ}) that $\k \geq 0$. Note that 
$\car = R (\cdot) R^\dagger$,
where 
$R : H_A \otimes H_B \rightarrow H_B \otimes H_A$ 
is defined by requiring
\[
R \ket{\varphi} \ket{\chi} = \ket{\chi} \ket{\varphi}  ,
\]
hence $\k' = R \k R^\dagger \geq 0$. Then by (\ref{kap'CJ}),
written as
\[
\k' 
= 
\sum_{kl} 
\ket{k^B}  \bra{l^B} \otimes \ce'( \s^{1/2} \ket{k^B} \bra{l^B} \s^{1/2} ),
\]
and Choi's criterion for complete positivity, we know that the map
$\mathcal{F} = \ce'( \s^{1/2} (\cdot) \s^{1/2})$ 
is completely positive. Hence its composition 
$\ce' = \mathcal{F} \circ \mathcal{D}$ 
with the c.p. map 
$\mathcal{D}=\sigma^{-1/2}(\cdot)\sigma^{-1/2}$, 
is indeed c.p. and therefore a channel, as required. 

If on the other hand we assume that $\ce$ is trace preserving, then since
the transpositions appearing in (\ref{duaal}) are taken with respect to
bases in which $\r$ and $\s$ are diagonal, while 
$\ce^\dagger(I_B) = I_A$ 
due to $\ce$ being trace preserving, we
have
\[
\ce'(\s) = \r.
\]

We summarize the central points when combining the separate assumptions above, that is, for a channel $\ce$:

\begin{theorem}
Consider a channel $\ce$ from $A$ to $B$ and a state $\r$ of $A$ such that the state $\s$ of $B$ given by $\s = \ce(\r)$ is an invertible operator. 
Then (\ref{kap'}) and (\ref{kap'CJ}) provide a uniquely determined channel $\ce'$ from $B$ to $A$. This channel is given by the formula (\ref{duaal}) in terms of transpositions with respect to orthonormal bases making $\r$ and $\s$ diagonal, and it satisfies $\ce'(\s) = \r$.
\end{theorem}

Although the AC dual was conceived in a more general and abstract setting, not at all connected to the notion of elementary transitions,
it turns out that in finite dimensions formula (\ref{duaal}) is also a formula for the AC dual
$\ce^{\AC}$ of $\ce$ with respect to $\r$ and $\s = \ce(\r)$ (refer to the Appendix).
Thus we have also shown the following key fact, which in the case of a channel is relevant to quantum
detailed balance in Subsection \ref{OndAfdDuaal&FB}.

\begin{theorem}
For any completely positive map $\ce$ from $L(H_A)$ to $L(H_B)$, and a state $\r$ of $A$ such that
$\s = \ce(\r)$ is invertible, we have
\begin{equation}
	\label{AC=duaal}
\ce^{\AC} = \ce'. 
\end{equation}
\end{theorem}

This attaches a simple physical meaning to the AC dual, as the reverse of $\ce$, which is not obvious from its original abstract definition.
It also provides an alternative method to derive properties of $\ce'$ from the known mathematical theory of the AC dual. 


\subsection{Quantum detailed balance}
\label{OndAfdDuaal&FB}

Let's return to the case of one system $A$, where $m = n$, 
and a channel $\ce$ from $A$ to itself as in Section \ref{AfdEofb}. 
We continue to assume that $\s = \ce(\r)$ is invertible to ensure we have the dual $\ce'$.
The goal here is to characterize ETDB as the requirement of the invariance
$\ce(\r) = \r$ along with the equation
\begin{equation}
	\label{eofbDeurDualiteit}
\ce' = \ce,
\end{equation}
and consequently to show that ETDB is equivalent to SQDB. As both ETDB and the condition $\ce' = \ce$ describe how the opposite of every elementary transition in $\ce$ also appears in $\ce'$, this characterization of ETDB is intuitive and also easy to derive.

To prove it, first suppose that $\ce$ satisfies ETDB relative to $\r$.
Consequently $\k' = \k$, 
and from invariance in Section \ref{OndAfdInv} we know that
$\s = \ce(\r) = \r$, 
allowing us to take $\ket{i^B} = \ket{i^A}$.
The CJ map giving $\k$ and
$\k'$ from $\ce$ and $\ce'$ in (\ref{CJweer}) and (\ref{kap'CJ}) respectively can be inverted, as we've assumed that $\s$ is invertible. 
Thus indeed $\ce' = \ce$.

Conversely, suppose that 
$\ce(\r) = \r$ and 
$\ce' = \ce$. 
By (\ref{CJweer}) and (\ref{kap'CJ}) it follows that $\k' = \k$, 
hence $\ce$ indeed satisfies ETDB relative to $\r$ by Theorem \ref{eofbKar}.

On the other hand, the AC dual allows us to define SQDB as
$\ce(\r) = \r$ along with
\begin{equation}
	\label{sfbDefAC}
\ce^{\AC} = \ce.
\end{equation}
See the Appendix and \cite[Example 5.2]{DSS} 
(in that paper the AC dual was denoted by a prime). 

By (\ref{AC=duaal}) we see that (\ref{sfbDefAC}) simply says $\ce' = \ce$.
Thus we have shown that ETDB is equivalent to SQDB. In particular, the
conceptual underpinnings of ETDB in this paper then applies to SQDB, which
serves to motivate the latter as a formulation of quantum
detailed balance. 
We conclude that SQDB is a natural quantum version of detailed balance in a classical Markov chain, though formulated in an abstract form in terms of the AC dual of $\ce$.

The abstract approach to quantum detailed balance through duality is itself made more intuitive because of the elementary transition point of view. 
We expressed ETDB in the duality framework as (\ref{eofbDeurDualiteit}),
but with the dual $\ce'$ in this case having a clear meaning in terms of every elementary transition in $\ce$ having been replaced by its opposite, making the connection to the original intuitively clear formulation Definition \ref{eofb} of ETDB quite transparent.


\subsection{The classical case}

Here we verify the meaning of the dual $\ce'$ as the reverse of
$\ce$ relative to $\rho$ in the classical case. 
Consider  a probability distribution 
$\rho=\left[
\begin{array}
[c]{rrr}
\rho_{1} & \cdots & \rho_{m}
\end{array}
\right]  $ 
over a set $A$ consisting of $m$ points, and a $m \times n$
transition matrix $\tau=[\tau_{jk}]$ from $A$ to a set $B$ consisting of $n$
points, meaning that $\tau_{jk}\geq0$ and
\[
\sum_{k=1}^{n}\tau_{jk}=1
\]
for all $j=1,...,m$. Defining a probability distribution 
$\sigma=\left[
\begin{array}
[c]{rrr}
\sigma_{1} & \cdots & \sigma_{n}
\end{array}
\right]  $ 
over $B$ by
\[
\sigma=\rho\tau,
\]
one can define the \emph{reverse} $\tau'$ of $\tau$ relative to
$\rho$, from $B$ to $A$, to accomplish the following:

Given the distributions $\rho$ and $\sigma$, the probability that the
transition $(k,j)$ takes place from $B$ to $A$ should be the same as the
probability for $(j,k)$ to take place from $A$ to $B$. That is, we require%
\[
\sigma_{k}\tau_{kj}^{\prime}=\rho_{j}\tau_{jk}
\]
for all $j$ and $k$.

Assuming
\[
\sigma_{k}>0
\]
for all $k$, this can indeed be achieved by setting
\[
\tau_{kj}^{\prime}=\frac{\rho_{j}}{\sigma_{k}}\tau_{jk}
\]
to give the transition matrix $\tau' = [\tau_{kj}']$.

On the other hand, embedding $\tau$ into a quantum representation as in
Subsection \ref{OndAfdInterpVanEO} by defining the channel $\ce$ by (\ref{klasInbed}),
we obtain the dual 
\[
\ce'(\ket{k^B}  \bra{l^B}  )
=
\delta_{kl}\sum_{j=1}^{m}
\frac{\rho_{j}}{\sigma_{k}}\tau_{jk}\ket{j^A}  \bra{j^A}  
=
\delta_{kl}\sum_{j=1}^{n}
\tau_{kj}'\ket{j^A}  \bra{j^A}  ,
\]
as is easily verified from the second equality of (\ref{E'formule}).
We thus see that the dual $\ce'$ indeed corresponds to the classical reverse of the
transition matrix through the quantum representation. Again this does not involve parity, which is what we turn to next.

\section{Parity}
\label{AfdPariteit}


To complete our approach to quantum detailed balance, at least for finite
dimensional state spaces, we also look at ETDB when some form of parity has to be taken into account. 
We'll denote ETDB in the case of parity by ETDB-$\cp$, where $\cp$ is the operation implementing parity, as will be described in detail in Subsection \ref{OndAfdEofbP}.

In the classical case a standard example of parity is in phase space, where 
momentum has to be negated as part of obtaining the reverse of a transition. 
In brief terms, the reverse of the transition of a particle from 
point $(q,p)$ to point $(q',p')$ 
in the phase space, is given by the transition of an identical particle from 
$(q',-p')$ to $(q,-p)$.
This change of sign is the parity in this case. 
See for example \cite[Section 5.3.4]{Gar} for a textbook discussion.

It is important to note that the transposition (with respect to our basis in which $\r$ is diagonal) as the reversing operation $\th$ in SQDB-$\th$ is not merely a choice in Subsection \ref{OndAfdDuaal&FB}.
It is the specific case that gives the equivalence of SQDB-$\th$ with ETDB.
In particular, this transposition should not be interpreted as being related to parity, 
since ETDB simply involves considering the opposite of every elementary transition rather than any form of parity.

The aim of this section is to contrast the case of ETDB with parity to
that without, and subsequently to continue casting light on the meaning of SQDB, and more generally SQDB-$\th$.
To do this, ETDB and SQDB will be shown to be equivalent even when parity is taken into account. As in Section \ref{AfdDuaal} this will be done using the dual of the channel, but now generalized to allow for parity.

In particular, 
the role of transposition as a reversing operation in previous literature on SQDB will be clearly separated from parity.
As part of this, in Subsection \ref{OndAfdOmkVanEOvsP} the distinction between the opposite direction of elementary transitions on the one hand, and parity on the other, will be emphasized by factorizing a general $\theta$ into transposition and a second operation $\cp$. It is this latter operation, rather than the transposition, which can then be identified as describing parity in SQDB-$\th$, because of the latter's equivalence to ETDB-$\cp$.


\subsection{ETDB-$\cp$}
\label{OndAfdEofbP}
To recall a quantum formulation of parity operations, we first state the
classical case more abstractly as follows, but on a finite number of points in
line with our framework thus far. 
This will give a clear path to incorporate parity into ETDB.

We consider any permutation $\pi$ of $\cx=\{1,...,m\}$ such that $\pi\pi = \id$, i.e., applying the permutation twice, has no effect. This permutation will act as the parity on the ``phase space'' $\cx$, and the reverse of the transition 
$i \rightarrow j$ 
under this parity will be taken to be 
$\pi(j) \rightarrow \pi(i)$. I.e., the reverse of $(i,j)$ under $\pi$ is $(\pi(j) , \pi(i))$. Detailed balance under this parity is then expressed as
\begin{equation}
	\label{klasFBmetP}
\rho_{i} \tau_{ij} = \rho_{\pi(j)} \tau_{\pi(j) \pi(i)},
\end{equation}
for all $i$ and $j$.

In terms of an orthonormal basis $\ket{1},...,\ket{m}$ for a Hilbert space $H$, we can give a quantum representation of $\pi$ as either a unitary or anti-unitary operator $P$ on $H$ defined through
\[
P\ket{i} = \ket{\pi(i)}
\]
which on operator level (or the Heisenberg picture) leads to 
\[
\cp(X) = PXP
\]
for any operator $X$ from $H$ to itself.
Given $\car$ from Subsection \ref{OndAfdDefEofb}, we'll write
\[
\cq = \car \circ (\cp \otimes \cp) = (\cp \otimes \cp) \circ \car,
\]
i.e.,
\[
\cq (X \otimes Y) = \cp(Y) \otimes \cp(X).
\]
In terms of (\ref{KwantVoorstVanKlasEO}) one easily verifies that
\[
\cq( \k_{(i,j)} ) 
=
\k_{ ( \pi(j) , \pi(i) )}.
\]

This enables us to adapt the approach that led to ETDB in Section \ref{AfdEofb}, to a framework allowing for parity. Returning to the quantum setting of that section, we now include parity as a unitary or anti-unitary operator $P$ on $H$ such that $P^2 = I$.
Again the parity $\cp$, as well as $\cq$, can be defined as above. 

Usually in the context of quantum detailed balance parity is taken to be time-reversal given by an anti-unitary $P$, but nothing in our development below will require it to be anti-unitary rather than unitary. 
Thus our more general view,
which will help clarify the role of $\th$ in parity in SQDB-$\th$ in Subsection \ref{OndAfdOmkVanEOvsP}.
The linear reversing operation for SQDB in Section \ref{AfdDuaal} for the case without parity turned out to be transposition. 
In Subsection \ref{OndAfdOmkVanEOvsP} it will be seen that more generally a linear reversing operation $\th$ is transposition composed with a $\cp$ given by a unitary $P$. To allow for anti-unitary $P$ as well, we'll consequently allow conjugate linear  $\th$.

\begin{example}
	Although we are working in a finite dimensional state space, we can give an example of parity in the form of momentum negation above, to illustrate our discussion of parity concretely in a very simple setting. 
	We use a commutation relation in unitary form for complementary observables in finite dimensions along the lines of \cite{S} as well as \cite[Chapter IV.D]{Weyl} (see \cite{Vour} for a review of various developments and references to applications).
	Having chosen our basis for $H$, we define unitary operators $U$ and $V$ on $H$ through the matrices
	\[
	U = 
	\left[
	\begin{array}
		[c]{cccc}
		0 &              &             & 1 \\
		1  & 0          &             & \\
		    & \ddots & \ddots & \\
		    &             & 1           & 0
	\end{array}
	\right]
	\qquad\mbox{and}\qquad
	V =
	\left[
	\begin{array}
		[c]{cccc}
		1 &    &              & \\
		   & r &              & \\
		   &   &  \ddots & \\
		   &   &              & r^{m-1}
	\end{array}
	\right]
	\]
	with $r = e^{2\pi i/m}$ and where the blank entries are zeros. Then 
	\[
	VU = rUV
	\]
	is the mentioned commutation relation, with $U$ and $V$ being connected by the Fourier transform for the cyclic group of $m$ elements (sometimes called the discrete or quantum Fourier transform) given by
	\[
	V = F^\dagger U F
	\]
	where
	$F = [F_{kl}]$ and $F_{kl} = r^{-kl}/m^{1/2}$.
	We think of $U$ and $V$ as a discrete position $q$ and momentum $p$ pair in unitary form
	\[
	U = e^{\frac{2\pi i}{m}q}
	\quad\mbox{and}\quad 
	V = e^{\frac{2\pi i}{m}p}.
	\]
	In particular, we are using a basis in which momentum is diagonal, and in our conventions $\r$ would have to be diagonal in this basis as well.
	In this context momentum parity can be implemented by complex conjugation (as defined in the Appendix) in the chosen basis, 
	$$P = C,$$
	as is easily verified by applying the resulting $\cp^\dagger$ to $V$, as the latter is from the observable algebra, not from the space of density matrices.
\end{example}

Returning to our general setting and arguing analogously to Subsection \ref{OndAfdDefEofb}, we introduce the following definition, which similar to Subsection \ref{OndAfdDefEofb}, contains the classical case (\ref{klasFBmetP}).
\begin{definition}
The channel $\mathcal{E}$ is said to satisfy ETDB relative to $\rho$ \emph{under the parity} $\cp$, or ETDB-$\cp$ relative to $\rho$ for short, if
there is a CJ decomposition (\ref{CJ-ontb}) of $\mathcal{E}$ relative to $\rho$ such that for
each $\alpha$ there is a $\beta$ satisfying
\[
	\mathcal{Q}(\kappa_{\alpha})
	=\kappa_{\beta}
	\text{ \ \ and \ \ }
	p_{\alpha}
	=
	p_{\beta}. 
\]
\end{definition}

The theory in this subsection and the next then proceeds essentially in parallel to the case without parity, with appropriate  changes to the arguments in the previous two sections.
Using the same argument as in Theorem \ref{eofbKar}, but in terms of the self-adjoint unitary operator 
$$
Q = R ( P \otimes P ) = ( P \otimes P ) R
$$ 
implementing $\cq$, instead of in terms of $R$, we find that ETDB-$\cp$ is satisfied if and only if
\[
\cq(\k)  = \k.
\]
\begin{sloppypar} 
	\noindent
	As in Subsection \ref{OndAfdOnafhVanCJ}, ETDB-$\cp$  is independent of the (necessarily) complete CJ-decomposition being used, but where completeness is now defined in terms of closedness under $\cq$ rather than $\car$.
\end{sloppypar}

For certain further developments, one additionally has to assume that $\rho$ commutes with $P$:
\begin{equation}
	\label{PrhoKom}
	P\rho = \rho P.
\end{equation}
This was not needed in the definition of ETDB-$\cp$ above, but it does become necessary for certain parts of the theory around ETDB-$\cp$ below.

For example, even classically, deriving invariance of the probability distribution under the transition matrix, 
\[
	(\rho\tau)_{j}
	=
	\sum_{i=1}^{n}\rho_{i}\tau_{ij}
	=
	\sum_{i=1}^{n} \rho_{\pi(j)} \tau_{\pi(j) \pi(i)}
	=
	\rho_{\pi(j)},
\]
requires the invariance 
$\r_{\pi(j)} = \r_j$
of the probability distribution under parity as a final step, translating to (\ref{PrhoKom}), $\cp(\rho) = \rho$, in the quantum case. 
Indeed, it is easily checked that the latter is used when adapting either argument from Subsection \ref{OndAfdInv} to prove that 
$$
\ce(\rho) = \rho
$$
when ETDB-$\cp$ is satisfied. This indicates that (\ref{PrhoKom}) is a reasonable assumption to make in ETDB-$\cp$.

In the next subsection it will be seen that to characterize ETDB-$\cp$ analogously to Subsection \ref{OndAfdDuaal&FB} in terms of a dual of the channel, will also require (\ref{PrhoKom}), although the basic theory for the particular dual does not.

\subsection{A dual channel incorporating parity}
\label{OndAfdDuaalMetP}

One can indeed extend the theory of dual channels in Subsection \ref{OndAfdDuaal} to incorporate parity. In the framework of that subsection, also consider parities $\cp_A$ and $\cp_B$ for $A$ and $B$ respectively, and set
\[
\cq = \car \circ (\cp_A \otimes \cp_B) = (\cp_B \otimes \cp_A) \circ \car,
\]
though not assuming that $\cp_A(\r) = \r$ or $\cp_B(\s) = \s$.

For any channel from $A$ to $B$ we still use $\k$ in (\ref{CJ}), but instead of $\k' = \car(\k)$ in (\ref{kap'}), we now consider it's reverse with parity included, which we denote by
\begin{equation}
	\label{kapP}
	\k^{\cp}
	=
	\cq(\k)
	=
	\sum_{ij}
	\rho_{i}^{1/2}\rho_{j}^{1/2} 
	\cp_B( \ce( \ket{i^A}  \bra{j^A} ) ) \otimes \cp_A( \ket{i^A}  \bra{j^A} ).
\end{equation}
Note that as opposed to $\k'$ in (\ref{kap'}), we are now working relative to the states
$\Tr_A(\k^\cp) = \cp_B(\s)$ 
and 
$\Tr_B(\k^\cp) = \cp_A(\r)$, 
instead of $\Tr_A(\k') = \s$ 
and $\Tr_B(\k') = \r$,
where $\Tr_A$ and $\Tr_B$ are the partial traces over $A$ and $B$ respectively.
Hence, when we set up the dual of $\ce$ corresponding to $\k^\cp$ in a way paralleling $\ce'$ in (\ref{kap'CJ}), we rather have to work relative to
$$
\sigma_{\cp} = \cp_B(\sigma) 
\ \ \text{and} \ \ 
\rho_{\cp} = \cp_A(\rho)
$$ 
and in terms of the bases $P_B \ket{j^B}$ and $P_A \ket{i^A}$ diagonalizing them.

Then the dual $\ce^\cp$ of $\ce$ corresponding to $\k^\cp$, is defined via
\[
	\label{kapPCJ}
	\k^\cp
	=
	\sum_{kl}
	\sigma_{k}^{1/2}\sigma_{l}^{1/2} ( P_B\ket{k^B}  \bra{l^B}P_B )
	\otimes
	\ce^\cp( P_B\ket{k^B}  \bra{l^B}P_B ), 
	\]
as $\ce'$ was via (\ref{kap'CJ}). Comparing this to (\ref{kapP}), but written in the form
\[
	\sum_{ij}
	\rho_{i}^{1/2}\rho_{j}^{1/2} 
	\cp_B \circ \ce \circ \cp_A ( P_A\ket{i^A}  \bra{j^A}P_A ) ) 
	\otimes 
	( P_A \ket{i^A}  \bra{j^A} P_A )
	\]
to have it in terms of the appropriate basis, we see that we have the same situation as in (\ref{kap'CJ}) versus (\ref{kap'}). This means that
\[
\ce^\cp = (\cp_B \circ \ce \circ \cp_A)',
\]
where it is implicit in our notation that this prime indicates the same dual as defined in Subsection \ref{OndAfdDuaal}, but relative to our current bases diagonalizing $\sigma_{\cp}$ and $\rho_{\cp}$.  From (\ref{duaal}) we conclude that
\[
\ce ^ \cp (Y)
=
\rho_{\cp}^{1/2}
\left[
{(\cp_B \circ \ce \circ \cp_A)}^\dagger
( \sigma_{\cp}^{-1/2} Y^{P_BT} \sigma_{\cp}^{-1/2}) 
\right]^ {P_AT}
\rho_{\cp}^{1/2}, 
\]
where $P_BT$ and $P_AT$ denote the transposes relative to the bases $P_B \ket{j^B}$ and $P_A \ket{i^A}$ and respectively. 
In addition we note that
$$
{(\cp_B \circ \ce \circ \cp_A)}^\dagger 
= 
\cp_A \circ \ce^\dagger \circ \cp_B
$$
as is easily verified, hence one can confirm from 
$X^{PT} = \cp ( \cp(X)^T )$,
for both $A$ and $B$, that
\begin{equation}
	\label{duaalMetP}
\ce^\cp = \cp_A \circ \ce' \circ \cp_B,
\end{equation}
with $\ce'$ the dual taken relative to the original bases $\ket{i^A}$ and $\ket{k^B}$ diagonalizing $\rho$ and $\sigma$.

Returning to Subsection \ref{OndAfdEofbP} with this dual in hand, one can then obtain a characterization of ETDB-$\cp$ in a similar way to Subsection \ref{OndAfdDuaal&FB}:

Assuming that 
$\cp(\rho) = \rho$, 
one can choose the orthonormal basis 
$\ket{1},...,\ket{m}$ 
such that $\r$ and $P$ are simultaneously diagonal, including for the case where $P$ is anti-unitary. Recall that $\r$ being diagonal is in force throughout in any case as part of our CJ decompositions. In particular we then know that 
$P\ket{i} = \pm \ket{i}$
for all $i$. It is then straightforward to adapt the argument in Subsection \ref{OndAfdDuaal&FB} to show that
the channel $\ce$ from $A$ to itself satisfies 
ETDB-$\cp$ relative to $\rho$
if and only if 
$\ce(\rho) = \rho$ and 
\begin{equation}
	\label{eofbPitvDualiteit}
\ce^\cp = \ce.
\end{equation}

Note that as with invariance in Subsection \ref{OndAfdEofbP}, we needed (\ref{PrhoKom}) for this characterization. This essentially just expresses the duals $\k^\cp$ and $\ce^\cp$ above in the same basis as that used in the initial CJ decomposition of $\ce$ that we are employing.

\subsection{SQDB-$\cp$}
\label{OndAfdOmkVanEOvsP}

The foregoing treatment of parity allows us to contrast it with the role of transposition 
in the definition of SQDB in Subsection \ref{OndAfdDuaal&FB}, for a system 
$A =B$. 
In this way we can clear up conceptual issues around SQDB and SQDB-$\th$ further, including explaining why this subsection's title refers to SQDB-$\cp$ rather than SQDB-$\th$.

As we have seen, SQDB  is equivalent to ETDB. In the latter, swapping the two copies of the system in the tensor product  simply enacts the opposite direction of elementary transitions, in step with the opposite direction of transitions in a classical Markov chain. The latter is not due to any parity related aspects like changing signs of momenta, but simply describes the two directions in which probability can flow between any two points in the classical case. Equality of these two flows for every pair of points being detailed balance. Analogously in the quantum case, but in terms of elementary transitions, as discussed in Section \ref{AfdEofb}. ETDB, and consequently SQDB, do not involve parity.

The AC dual in the definition (\ref{sfbDefAC}) of SQDB itself only involves the opposite direction for each elementary transition, not parity, as $\ce^{\AC} = \ce'$. 
We are using a basis making the state $\rho$ of the system diagonal to set up the CJ map and CJ decompostion of $\ce$.
Transposition in this basis appears in the AC dual, and that's how it enters the definition of SQDB.
Any matters associated to parity will be seen to be distinct from this transposition. 

To see this clearly, let's return to SQDB-$\theta$ for a general reversing operation $\theta$ written in the form
\[
\theta(X) = \Theta X^\dagger \Theta,
\]
where $\Theta$ is an anti-unitary or unitary operator such that 
$\Theta^2 = I$. 
The anti-unitary case corresponds to the literature on SQDB-$\th$.

To avoid any confusion regarding different conventions, we mention that in Section \ref{AfdDuaal} the transposition can be  written in the form
\begin{equation}
	\label{TitvC}
X^T = C X^\dagger C,
\end{equation}
where, as in the Appendix, $C$ is the anti-unitary operator which takes the complex conjugate of each component of a vector when expressed in our chosen basis.
However, one can use an alternative convention for reversing operations in SQDB-$\theta$, and rather write (\ref{TitvC}) in the form
$
\th(X) = CXC,
$
expressing SQDB in terms of this instead, and indeed extend this to the more general case $\th(X) = \Theta X \Theta$ for some anti-unitary operator $\Theta$. 
But since having the transposition itself as an example of a reversing operation in SQDB-$\th$ is convenient, we rather use the convention $\Theta X^\dagger \Theta$, as in some of the literature (like \cite{FR}). 
To disentangle the roles of $\th$ and parity $\cp$, though, it is of value to allow for unitary $\Th$ as well.

As in the usual definition of SQDB-$\th$ in the literature, we assume that 
\[
\th(\r) = \r,
\]
i.e., $\Theta\rho = \rho\Theta$,
and we choose our orthonormal basis to simultaneously diagonalize both $\rho$ and $\Theta$. 
In a moment we'll see that under these assumptions $\theta$ factorizes into transposition and a uniquely determined parity $\cp$ given by
$\cp(X) = PXP$
for a diagonal 
unitary or anti-unitary $P$, when $\Th$ is anti-unitary or unitary, repectively, 
and that SQDB-$\theta$ is equivalent to ETDB-$\cp$. 
This will clearly delineate the role of transposition versus that of parity.

Indeed, as explained at the end of the appendix, this $P$ is given by what is referred to as $\bar{\Theta}$ there, giving the simple factorization
\begin{equation}
\label{thetaFakt}
\theta(X) = \cp(X^T) = \cp(X)^T,
\end{equation}
with SQDB-$\theta$ being expressed by $\ce(\rho) = \rho$ along with the equation
\begin{equation}
	\label{sfbDefP}
	\cp \circ \ce^{\AC} \circ \cp = \ce
\end{equation}
extending $\ce = \ce^{\AC}$ in the definition (\ref{sfbDefAC}) of SQDB without parity involved. This equation is equivalent to
\[
\ce^\cp = \ce,
\]
where we have also used (\ref{AC=duaal}) and the results of Subsection \ref{OndAfdDuaalMetP}.
This tells us that SQDB-$\theta$ is precisely ETDB-$\cp$.

The conventional case of a parity operation where $P$ is anti-unitary, corresponds to unitary $\Th$, which is why we allowed for this option in $\th$ above.
For example, when $P = C$, we have $\th(X) = X^\dagger$.

Condition (\ref{sfbDefP}) for SQDB-$\th$ above 
(together with $\ce(\rho) = \rho$) 
is the usual definition in the literature, but written in our notation and allowing for unitary $\Th$. (Also see the Appendix.) In this form it is clear that SQDB-$\th$ can equivalently be viewed as depending on $\cp$ rather than $\th$.

The equivalence between SQDB-$\theta$ and ETDB-$\cp$ tells us that any parity involved in SQDB-$\theta$ is isolated in $\cp$ obtained by factorizing $\theta$ as in (\ref{thetaFakt}), rather than being given by transposition. 
This is simply because in our approach to ETDB-$\cp$ we have argued that $\cp$ describes parity. 

Keep in mind that all of this is in terms of our chosen basis, in which both $\rho$ and $\Theta$ are diagonal.
In particular, when $\theta$ is taken as the transposition in this basis, there is no parity involved in the resulting SQDB condition.
This transposition is part of the underlying duality used to define SQDB, with or without parity. This duality, with the transposition built in,
emerges from the elementary transition approach and ETDB, again with or without parity;
recall (\ref{duaal}), (\ref{eofbDeurDualiteit}), (\ref{duaalMetP}) and (\ref{eofbPitvDualiteit}).
Hence this transposition should not be viewed as an externally added operation to describe parity. 
Indeed, from this point of view SQDB-$\th$ is technically a misnomer. It should rather be called SQDB-$\cp$, namely standard quantum detailed balance with respect to a parity operation $\cp$, and should be defined from this perspective from the outset.

That is, given a parity $\cp$, in our  conventions SQDB-$\cp$ of $\ce$ relative to $\r$ is defined by $\ce(\rho) = \rho$ together with (\ref{sfbDefP}).

%
%

\section{Concluding remarks}

We have provided a foundation for quantum detailed balance based on what we termed elementary transitions. The latter are as direct an analogue of transitions (between pure states) in a classical Markov chain as seems to be allowed by quantum mechanics. They are simple to formulate as pure states of the composite of two copies of the system, but are no longer in general transitions between pure states of the system. 
This is more intuitive than previous approaches to quantum detailed balance, through the analogy with the classical case as summarized by 
$\rho_{i} \tau_{ij} = \rho_{j} \tau_{ji}$, 
though still being purely quantum mechanical. 
It clarifies conceptual aspects of quantum detailed balance, while also leading to a rigorous mathematical setting connecting to previous mathematical work related to quantum detailed balance. 

We expect that the elementary transition point of view can also be of value in other aspects of quantum dynamics.
Reversing quantum dynamics is one theme already suggested by our approach in this paper.

In this paper we essentially built the theory by starting with quantum channels, and obtaining elementary transitions from them. 
Instead of taking channels as the starting point, one could explore using elementary transitions or their collective description $\k$ as the primary notion to express dynamics. In some respects this would be more general than channels, including in the context of detailed balance. 
Elements of this approach was in fact used in Section \ref{AfdInlVb}, 
and its potential value illustrated in the context of invariance in Subsection \ref{OndAfdInv},
but it was not explored systematically in this paper.

\section*{Appendix: The Accardi-Cecchini and KMS duals, and SQDB-$\th$}



We briefly outline the abstract theory of the Accardi-Cecchini (AC) and KMS duals in finite
dimensions to obtain concrete formulas for them and to define SQDB-$\th$.
We thus explain the definition
$\ce^{\AC}=\ce$ 
for SQDB used in Subsection
\ref{OndAfdDuaal&FB}, and similarly for SQDB-$\th$ in Subsection \ref{OndAfdOmkVanEOvsP}.
In addition, technical points regarding reversing operations $\th$ themselves are discussed for the benefit of Subsection \ref{OndAfdOmkVanEOvsP}. 
This includes a simple factorization of a reversing operation into transposition with respect to the chosen basis in which $\rho$ is diagonal, and a parity operation given by a unitary operator, or anti-unitary if we allow sufficiently general $\th$.

The AC dual was originally studied in \cite[Proposition 3.1]{AC}.
The KMS dual was subsequently systematically developed by Petz \cite{Pet} in a more general setting to cater for weights rather than just states (also see \cite{Pet86, Pet88}). 
However, for states (covering our quantum mechanical needs) the KMS dual is mathematically essentially the AC dual expressed or represented in an alternative way, and it already appeared in \cite{AC} as well, being called the bidual there. 
The KMS dual is also referred to as the standard dual or Petz recovery map. 
It has played an important role in work on reversibility of dynamics; see for example \cite{Pet88, Ts1}).

Consider a positive map 
$$\mathcal{E}:L(H_{A})\rightarrow L(H_{B}).$$
Use the notation
\[
E=\mathcal{E}^{\dagger}
\]
where $\mathcal{E}^{\dagger}:L(H_{B})\rightarrow L(H_{A})$ is defined through
$\Tr(X\ce^{\dagger}(b)) = \Tr(\ce(X)b)$. 
Here we write $X$ as an
element of $L(H_{A})$ viewed as the space containing the density matrices, and
$b$ as an element of $L(H_{B})$ viewed in the dual sense as an observable
algebra. As the literature mentioned above is in the context of operator
algebras and positive maps on them, they are working in terms of $E$
rather than $\mathcal{E}$, and the duals are defined in this context, but in a
far more general von Neumann algebraic framework than just the finite
dimensional case. To keep track of these dual points of view (we could say the
Schr\"{o}dinger and Heisenberg pictures respectively), we'll also write $a$
and $b$ for elements of any $L(H)$ viewed as an observable algebra, and $X$
and $Y$ for elements of $L(H)$ viewed as the space containing the density
matrices. To emphasize the role of $E$, we'll write
\[
E:\mathcal{B}\rightarrow\mathcal{A}
\]
where $\mathcal{A}$ and $\mathcal{B}$ denote $L(H_{A})$ and $L(H_{B})$
respectively, but when viewed as observable algebras.

\subsection*{The AC dual}
Define the Hilbert-Schmidt inner product on any $L(H)$ by
\[
\left\langle X,Y\right\rangle = \Tr(X^{\dagger}Y)
\]
and define $\pi$ and $\pi'$ by
\[
\pi(a)X = aX \text{ \ \ and \ \ } \pi'(a)X = Xa^{T}
\]
for all $a$ and $X$ in $L(H)$, in terms of the transposition $a^{T}$ of $a$ in a chosen orthonormal basis.
The prime in $\pi'$ in this case is in reference to the commutant of a
von Neumann algebra in standard form in the abstract theory, 
which in this concrete finite dimensional  treatment boils down to the right multiplication by $a^T$. 
The standard form in this case is effectively provided by the chosen basis. (The theory of standard forms of von Neumann algebras was originally developed in \cite{Ar74, Con74, Ha75}.)
The $\AC$ dual
\[
E^{\AC} : \mathcal{A}\rightarrow\mathcal{B}
\]
of $E$ is defined as the map satisfying the following:
\begin{equation}
	\label{ACdef}
\left\langle 
\s^{1/2} , \pi(b) \pi'( E^{\AC}(a)) ) \s^{1/2}
\right\rangle 
=
\left\langle 
\rho^{1/2},\pi(E(b))\pi'(a)\rho
^{1/2}
\right\rangle
\end{equation}
with $\rho$ and $\sigma$ as well as our chosen bases as in Subsection \ref{OndAfdDuaal}. 
In particular we assume that $\s$ is invertible and
$$
\sigma = \ce(\rho).
$$
By straightforward manipulations it then follows that
\[
E^{\AC}(a)
=
\sigma^{-1/2}
E^{\dagger}(\rho^{1/2}a^{T}\rho^{1/2})
^ T
\sigma^{-1/2}
\]
where $E^{\dagger}=\mathcal{E}$ and with the transpositions taken in the chosen bases. 
This dual expressed in terms of the
Scr\"{o}dinger picture, i.e., in terms of 
$\ce^{\AC}=(E^{\AC})^{\dagger}$, 
is then easily verified to be given by
\[
\ce^{\AC}(Y)
=
\r^{1/2}
\ce^{\dagger}(\s^{-1/2} Y^T \s^{-1/2})
^ T
\r^{1/2}
\]
for all $Y$ in $L(H_{B})$. This is exactly the formula obtained for
$\ce'$ in (\ref{duaal}). So indeed 
$\ce^{\AC} = \ce'$.

The reader is also referred to \cite[Theorem 2.5]{DS} for a brief review of the AC dual in a more general setting and in a form convenient for our purposes.

\subsection*{The KMS dual}
As for the KMS dual 
$E^{\KMS} : \mathcal{A} \rightarrow \mathcal{B}$ 
of $E$, it
can be defined in terms of the representation $D(\pi(b))=\pi(E(b))$ of $E$ as
the composition of maps
\[
D^{\KMS} = j_{B} \circ D^{\AC} \circ j_{A}
\]
\begin{sloppypar} 
\noindent
where $D^{\AC}$ represents 
$E^{\AC} : \mathcal{A} \rightarrow \mathcal{B}$ 
in the corresponding way, that is
$D^{\AC}(\pi'(a))=\pi'(E^{\AC}(a))$, 
and $E^{\KMS}$ has to be
recovered from the representation 
$D^{\KMS}(\pi(a)) = \pi(E^{\KMS}(a))$. 
The $j$'s in this formula arise from the so-called
modular conjugation of Tomita-Takesaki theory which map between a von Neumann algebra in standard from and its commutant (see \cite{T70} for an early
source and \cite{T2} for a later textbook treatment).
\end{sloppypar}

Indeed, the natural mathematical form of $E^{\AC}$ is in terms of the commutants of the algebras in their standard form, which are visible in (\ref{ACdef}) above in the expressions
$\pi'(a)$ and $\pi'(E^{\AC}(a)))$.
These expressions translate to elements of the commutants (in standard form) of $\ca$ and $\cb$ respectively.
The KMS dual is the mathematically natural way of carrying the AC dual over to the observable algebras themselves, and in this sense is essentially just a different way of expressing the AC dual.
In fact, exactly this construction of the KMS dual already appeared in \cite[p. 254]{AC} under the name ``bidual map''.
Petz's more general definition of $E^{\KMS}$ through a duality relation in \cite{Pet} (also see a more specialized later version \cite{Pet88}), is equivalent to [1]'s construction where their assumptions overlap, in particular for the case of channels and when working relative to states (instead of weights).

In our setup the $j$'s are given
by transposition with respect to the chosen bases, namely
\[
j_{A}(\pi(a)) = \pi'(a^{T}) \text{ \ \ and \ \ } j_{B}(\pi'(b)) = \pi(b^{T}).
\]
Consequently,
\[
E^{\KMS}(a) = E^{\AC}(a^{T})^{T}.
\]
This then gives a formula for $E^{\KMS}$ from that of $E^{\AC}$ above, and due to
$E^{\AC} = E'$
by (\ref{AC=duaal}) of course also relates $E^{\KMS}$ to $E'$. Correspondingly for $\ce^{\KMS}$. 
This formula,
$$
\ce^{\KMS}  (Y)
= 
\ce^{\AC}(Y^{T})^{T}
=
\r^{1/2}
\ce^{\dagger}(\s^{-1/2} Y \s^{-1/2})
\r^{1/2},
$$
from that of $\ce^{\AC}$ above, appears for example in the textbook treatment \cite[Section 12.3]{Wi}, being called the Petz recovery map there.

\subsection*{SQDB-$\th$}
With 
$H = H_A = H_B$, 
and under the assumption
$$\ce(\rho) = \rho,$$
SQDB-$\theta$ of $E$ relative to $\r$ can be defined as 
\begin{equation}
    \label{sfbDef}
\theta \circ E^{\KMS} \circ \theta = E
\end{equation}
for a reversing operation $\th$ such that $\th(\r) = \r$ (see for example \cite[Definition 1]{FR} or \cite[Section 6]{DS}). 
Abstractly in our conventions, the reversing operation $\theta$ in the literature is taken as a linear operation preserving the hermitian adjoint but swapping the product (an anti-$*$-automorphism), though we'll turn to a more concrete expression for it shortly.
In our current setting, and with $\theta$ taken as
transposition, (\ref{sfbDef}) can be written as
\[
E^{\AC} = E,
\]
as has also been pointed out in \cite[Example 5.2]{DSS}. I.e., 
$$
\ce^{\AC} = \ce
$$
in the Schr\"{o}dinger view.  
Keep in mind that this, along with the assumption 
$\ce(\rho) = \rho$ 
above, is what is simply referred to as SQBD in Section \ref{AfdDuaal} (see in particular Subsection \ref{OndAfdDuaal&FB}).

Note that the transposition appears as a natural part of the $\AC$ dual in the chosen bases, rather than as an externally introduced operation.
In  the KMS dual they are cancelled by its mathematically natural definition.
This is a simple conceptual point relevant to clarifying some aspects of parity when added to SQDB in Section \ref{AfdPariteit}.
In particular, to help understand that we should separate this transposition 
from parity.

The remainder of this appendix treats some mathematical points connected to parity, in particular related to the separation mentioned above. For a general reversing operation $\theta$ of the form
\[
\theta(X) = \Theta X^\dagger \Theta,
\]
as appearing in the usual formulation of SQDB-$\th$ \cite{FU}, where $\Theta$ is an anti-unitary operator with
$\Theta^2 = I$ and
\[
\Th\rho = \rho\Th,
\]
we see from (\ref{sfbDef}) that SQDB-$\theta$ is similarly expressed as 
\[
\theta \circ T \circ E^{\AC} \circ T \circ \theta = E
\]
where we have written $T$ for transposition, i.e., 
$$T(a) = a^T.$$

However, in keeping with the discussion in Subsection \ref{OndAfdOmkVanEOvsP}, we'll also allow unitary $\Th$, making $\th$ conjugate linear. This therefore extends the usual definition of SQDB-$\th$. We need this extension to handle parity properly in Subsection \ref{OndAfdOmkVanEOvsP}.

Now define $\bar{\Th}$ as the unitary (when $\Th$ is anti-unitary) or anti-unitary (when $\Th$ is unitary) operator given by
\[
\bar{\Th} = C \Th = \Th C
\]
where $C$ is component-wise complex conjugation in the chosen basis.
That is, $C$ is the anti-unitary operator defined as the conjugate linear operator satisfying
$$C\ket{i} = \ket{i}.$$
Here $\ket{1},...,\ket{m}$ are chosen to form an orthonormal basis for $H$ with respect to which both $\rho$ and $\Theta$ are diagonal as in Subsection \ref{OndAfdOmkVanEOvsP}, and with respect to which the transposition is being taken. This is of course the basis to be used in the CJ map and CJ decomposition of $\ce$.
In Subsection \ref{OndAfdOmkVanEOvsP} it is seen that $\bar{\Th}$ serves as the parity operator $P$ discussed in Subsection \ref{OndAfdEofbP}. From the definition above, $\bar{\Th}$ simply changes $\Th$ from conjugate linear to linear, or vice versa, but still having the same values as $\Th$ on the basis.

We clearly have
\[
\Th = \bar{\Th}C = C\bar{\Th}.
\]
Consequently
\[
\theta = \bar{\theta} \circ T = T \circ \bar{\theta},
\]
that is, $\theta$ factorizes as the transposition and an operation $\bar{\theta}$ defined by
\[
\bar{\th} (X) = \bar{\Th} X \bar{\Th}.
\] 
Clearly $\bar{\th}$ is uniquely determined by $\th$, since the factorization is equivalent to
$\bar{\theta} = \theta \circ T = T \circ \theta$ by simply composing it with $T$.
The generalized SQDB-$\theta$ condition above (where both anti-unitary and unitary $\Th$ are allowed) can now be written as
\[
\bar{\th} \circ E^{\AC} \circ \bar{\th} = E,
\]
or equivalently in the Schr\"odinger picture as
\[
\bar{\th} \circ \ce^{\AC} \circ \bar{\th} = \ce,
\]
of course still under the assumption 
$\ce(\rho) = \rho$.

As seen above, the transposition is a natural part of $\ce^{\AC}$, which is already present in SQDB without any parity involved.
Hence at least mathematically SQDB-$\th$ just depends on $\bar{\th}$, making the alternative name SQDB-$\bar{\th}$ more appropriate. 
In Subsection \ref{OndAfdOmkVanEOvsP} the same conclusion is reached from the perspective of ETDB under parity.

\section*{Acknowledgements}
We thank the referees for their feedback and suggestions.

\end{document}